\documentclass[twoside,11pt]{article}

\usepackage[abbrvbib, preprint]{jmlr2e}

\usepackage{paralist, amsmath, amssymb, color, graphicx, algorithm, algpseudocode, float, floatflt, cite, mathtools, bm}
\usepackage{tikz, pgfplots}
\pgfmathdeclarefunction{gauss}{2}{%
  \pgfmathparse{1/(#2*sqrt(2*pi))*exp(-((x-#1)^2)/(2*#2^2))}%
}
\usepackage{enumitem, bigstrut}
\usepackage{booktabs, multirow, threeparttable}
\usepackage{setspace}
\usepackage{mathabx}
\usepackage{thmtools, thm-restate, cleveref}
\usepackage{wrapfig}
\usepackage{xcolor}

\pgfplotsset{compat=1.18}
\hypersetup{hidelinks}

\newcommand{\ex}{\mathbb{E}}
\newcommand{\var}{\text{Var}}
\newcommand{\cov}{\text{Cov}}
\newcommand{\sdp}{{s_{dp}}}

\usepackage{lastpage}

\ShortHeadings{Particle Filter for Bayesian Inference on Privatized Data}{Chen, Sanghi and Awan}
\firstpageno{1}
\begin{document}

\title{Particle Filter for Bayesian Inference on Privatized Data}

\author{\name Yu-Wei Chen \email chen4357@purdue.edu \\
       \addr Department of Statistics\\
       Purdue University \\
       West Lafayette, IN 47907, USA
        \AND
       \name Pranav Sanghi \email psanghi04@gmail.com  \\
       \addr Department of Computer Science\\
       Purdue University \\
       West Lafayette, IN 47907, USA
       \AND
       \name Jordan Awan \email jaa557@pitt.edu  \\
       \addr Department of Statistics\\
       University of Pittsburgh \\
       Pittsburgh, PA 15260, USA}


\maketitle

\begin{abstract}
    Differential Privacy (DP) is a probabilistic framework that protects privacy while preserving data utility. To protect the privacy of the individuals in the data set, DP requires adding a precise amount of noise to a statistic of interest; however, this noise addition alters the resulting sampling distribution, making statistical inference challenging. One of the main DP goals in Bayesian analysis is to make statistical inference based on the private posterior distribution. While existing methods have strengths in specific conditions, they can be limited by poor mixing, strict assumptions, or low acceptance rates. We propose a novel particle filtering algorithm, which features (i) consistent estimates, (ii) Monte Carlo error estimates and asymptotic confidence intervals, (iii) computational efficiency, and (iv) accommodation to a wide variety of priors, models, and privacy mechanisms with minimal assumptions. We empirically evaluate our algorithm through a variety of simulation settings as well as an application to a 2021 Canadian census data set, demonstrating the efficacy and adaptability of the proposed sampler.
\end{abstract}


\section{Introduction}

Differential Privacy (DP; \citealp{dwork2006calibrating}) is a rigorous probabilistic framework designed to protect the privacy of individuals in a data set by adding noise to the queries made on the data. It has become widely adopted in academia (\citealp{vadhan2017complexity} from Harvard OpenDP), industrial companies (\citealp{erlingsson2014rappor}, \citealp{abadi2016deep} from Google Research; \citealp{dwork2006differential}, \citealp{ding2017collecting} from Microsoft Research), and governments (\citealp{Abowd20222020} from U.S. Census Bureau). Such privacy-protection mechanisms inevitably come with many challenges in conducting valid statistical inference on privatized data. One of the major challenges is that, given a prior model $\theta \sim \pi_0 (\cdot)$ for the target parameters, a data model $x \sim f(\cdot \mid \theta)$ for the sensitive data, and a privacy mechanism $\sdp \mid x \sim m(\cdot \mid x)$, the marginal private posterior distribution $\pi(\theta \mid \sdp) \ \propto \ \pi_0 (\theta)\int f(x \mid \theta) m(\sdp \mid x) \ dx$ is typically intractable, as it requires integrating over all possible databases $x$. The doubly intractable computation \citep{murray2006mcmc} cannot be directly solved by traditional Markov Chain Monte Carlo (MCMC) methods, such as the Metropolis-Hastings algorithm, but requires a more careful approach.

Under the Bayesian regime, MCMC methods have arisen to exactly sample from or at least to target sampling from the private posterior $\pi(\theta \mid \sdp)$ (\citealp{Bernstein2018} $\&$ \citeyear{bernstein2019differentially}; \citealp{gong2022exact}; \citealp{ju2022data}). (i) \citeauthor{Bernstein2018}  (\citeyear{Bernstein2018} $\&$ \citeyear{bernstein2019differentially}) are limited to certain priors, data models, and mechanisms. (ii) \citet{gong2022exact} provides i.i.d. samples from the exact private posterior but can suffer from low acceptance rates.   (iii) While \citet{ju2022data} gives a more generic missing-data sampling approach, it is optimized for conjugate prior distributions. In other words, along with the usual mixing issues in MCMC methods, it can be inefficient when non-conjugate prior distributions are used. Even though conjugate priors are analytically convenient, they are known to be non-robust to mis-specifications \citep{Jim1985}; for more robust Bayesian inference, \citet{berger1994overview} recommend using heavier-tailed prior distributions than those arising from standard conjugate priors.

Particle filtering  (PF), also known as Sequential Monte Carlo (SMC), is a powerful method of obtaining approximate samples from a posterior distribution, assuming the prior is proper and both the prior and data model can be easily sampled. The algorithm begins with a set of particles drawn from the prior distribution, which are then iteratively perturbed and filtered. Then, it evolves recursively through the following three steps: propagation (mutation), reweighting (correction), and resampling (selection). Through repeated iterations, the particle system gradually generates samples that approximate the target posterior distribution. 

Under mild conditions---generally weaker than (geometric) ergodicity assumptions required in MCMC---PF exhibits favorable particle behavior, including strong convergence and stability \citep{chopin2004}, while naturally accommodating parallel computation. Moreover, leveraging its inherent sequential construction at no additional computational cost, PF provides an estimator of the marginal likelihood (model evidence) via the product of incremental normalizing-constant ratios, thereby enabling Bayesian model comparison and selection. However, the preservation of these properties under differential privacy is not automatic and requires careful design of the privacy-aware procedure, along with theoretical analysis.

\textbf{Our Contributions.} 
In this paper, we develop a custom particle filtering algorithm to sample from the private posterior distribution $\pi(\theta,x \mid \sdp)$. To tailor the algorithm to the DP setting, we introduce two customizations: (i) an artificial schedule on the privacy parameter $\epsilon$, and (ii) a rejection step based on \citet{gong2022exact} that incorporates the privacy schedule. The schedule has a tempering effect on the intermediate target distribution, and the rejection step helps to efficiently target the private posterior distribution, avoiding the bias introduced by approximate Bayesian computing (ABC) approaches. Following the generic structure of particle filtering, our method
\begin{itemize}
    \item is widely applicable and only requires the ability to sample from the prior and data model, and the ability to evaluate the densities of the prior and the privacy mechanism,
    \item uses a novel tempering of the privacy budget $\epsilon$, in place of the usual ABC schedule,
    \item allows for parallelization across the particles in each cycle,
    \item results in consistent estimates, the central limit theorem (CLT), and accurate estimates of the Monte Carlo standard errors (MCSE), giving an asymptotic confidence interval (CI) for the true posterior quantities, and
    \item yields an estimate of the marginal likelihood for the sensitive data.
\end{itemize}

Importantly, these estimates are nontrivial to derive in the DP setting and may be compromised without the specific structural design of our proposed algorithm (see Section~\ref{sec: particle filters for DP} for further elaboration). To evaluate our method, we apply it to location-scale normal, linear regression, and Gaussian mixture models to demonstrate its improved performance on computation efficiency compared to the data augmentation MCMC of \citet{ju2022data} and the rejection ABC algorithm of \citet{gong2022exact}, which will be referred to hereafter as DP-DA-MCMC and DP-Reject-ABC, respectively. Additionally, we apply our method to 2021 Canadian Census data to perform a DP logistic regression via objective perturbation, showcasing its practical application.

\textbf{Organization.} 
The remainder of the paper is organized as follows. Section~\ref{sec: background} reviews the definition of $\epsilon$-differential privacy and key properties. We also introduce the structure and statistics of a typical particle filtering method (not specific to the privatized data setting). Section~\ref{sec: particle filters for DP} first introduces the difference between our algorithm and a typical particle filtering scheme. Section~\ref{subsec: algorithm} presents our main algorithm along with its assumptions and designs. In Section~\ref{subsec: consistency} and Section~\ref{subsec: clt}, we prove our estimator has strong consistency and asymptotic normality. In Section~\ref{subsec: mcse and asymptotic ci} and Section~\ref{subsec: evidence}, we derive the asymptotic confidence interval for the true parameter and marginal likelihood of the sensitive data. In Section~\ref{sec: simulation}, we compare our method to the DP-DA-MCMC and DP-Reject-ABC in various simulation studies. In Section~\ref{sec: data analysis}, we analyze the model parameters by applying our algorithm to a logistic regression via objective perturbation on the 2021 Canadian census data set. Section~\ref{sec: conclusions} discusses the implications and limitations of our work. Proof, technical details, and additional results are deferred to the appendices. The source code used to implement the simulation study and data analysis is available at \url{https://github.com/psanghi04/DP-PF}.

\textbf{Related Work.} 
Common strategies for statistical inference on privatized data can roughly be categorized into two major perspectives: asymptotic frequentist and Bayesian. The first employs traditional statistical asymptotics to approximate the sampling distribution of differentially private statistics, arguing that the privacy noise is negligible compared to sampling errors in large samples. However, these approximations can fall short in finite samples and require specific methods to account for privacy effects (\citealp{smith2011}; \citealp{cai2021cost}; \citealp{wang2018statistical}) and are tailored to specific privacy mechanisms (\citealp{gaboardi2016differentially}; \citealp{gaboardi2017local}). 

The second category emphasizes the marginal likelihood of model parameters, which involves high-dimensional integrals over unobserved confidential databases that are analytically tractable only in simple cases (\citealp{awan2018ump}, \citeyear{awan2020differentially}). To address these challenges, various Markov Chain Monte Carlo (MCMC) methods have been proposed for specific privacy mechanisms and data models, including approaches for exponential random graph models, exponential family models, and generalized linear models (\citealp{karwa2017sharing}; \citealp{Bernstein2018}, \citeyear{bernstein2019differentially}; \citealp{schein2019locally}; \citealp{kulkarni2021differentially}). Recent developments include \citet{gong2022exact}, who showed that approximate Bayesian computation (ABC) can provide samples consistent with the marginal likelihood and Bayesian posterior for certain differentially private statistics. \citet{ju2022data} proposed a method for Bayesian inference that integrates existing samplers for non-private data while accounting for the privacy mechanism, although it relies on strong assumptions regarding geometric ergodicity and may be inefficient in some cases.

In addition, some approaches do not fit into the two categories. \citet{awan2025simulation} developed a simulation-based finite-sample inference strategy for privatized data, providing accurate confidence sets at the cost of increased computation. Other frequentist approaches use different approximations, such as parametric bootstrapping \citep{ferrando2022parametric,wang2025optimal} and non-parametric bootstrapping \citep{wang2025differentially}. There is also an approach based on Bayesian asymptotics, but this is currently limited to additive noise mechanisms and summation-based summaries without clamping \citep{awan2026large}.

\section{Background} \label{sec: background}

In this section, we review the necessary background and set the notation for the paper.

We let $x = (x_1,\ldots,x_n)$ be a database in $\mathcal{X}^n$ and under a data model $f(\cdot \mid \theta)$, where $\theta$ is the parameter of interest. In the DP framework, $\sdp$ is denoted as the observed, privatized statistic, and $m_{\epsilon} (\cdot \mid x)$ is the density of a privacy mechanism conditional on confidential data $x$, where $\epsilon$ is the privacy parameter. In a Bayesian setting, we assume a prior $\pi_0 (\cdot)$ on $\theta \in \Theta$ and the Bayesian model can be summarized as follows:
\begin{align*} 
\begin{split}
    \theta &\sim \pi_0 (\theta), \\
    x \mid  \theta &\sim f(x \mid \theta), \\
    \sdp \mid x &\sim m_{\epsilon} (\sdp \mid x).
\end{split}
\end{align*}
To infer $\theta$ based on $\sdp$, we are interested in the private posterior 
\[\pi(\theta \mid \sdp) \ \propto \ \pi_0 (\theta) \int_{\mathcal{X}^n} f(x \mid \theta) m_{\epsilon} (\sdp \mid x) \ dx,\]
or more generally, in
\begin{align*}
    \pi(\theta, x \mid \sdp) \ \propto \ \pi_0 (\theta) f(x \mid \theta) m_{\epsilon} (\sdp \mid x).
\end{align*}

\subsection{Differential Privacy}

Differential privacy (DP), introduced in \citet{dwork2006calibrating}, provides a framework to quantify the privacy risks associated with an algorithm and offers techniques to design privacy mechanisms that control privacy loss. To protect sensitive information in the data set, DP methods require adding a proper amount of randomness to the summary statistic with the goal of ensuring that adversaries cannot easily determine whether a specific individual was part of the data set. Meanwhile, the utility of the private output is ideally still informative for many population-level statistics.

\begin{definition}[Privacy Mechanism]
    Given a set $\mathcal{X}^n$, which is the collection of all possible databases, a \emph{privacy mechanism} $\mathcal{M}$ is defined as a set of probability measures $\{ M_x \mid x \in \mathcal{X}^n \}$ which take values on a common 
    measurable space $\mathcal{Y}$.
\end{definition}

\begin{definition}[$\epsilon$-Differential Privacy: \citealp{dwork2006calibrating}]
    Given a set $\mathcal{X}^n$ and the Hamming distance $d_H(\cdot, \cdot)$ on $\mathcal{X}^n \times \mathcal{X}^n$, a privacy mechanism $\mathcal{M}$ satisfies $\epsilon$-differential privacy ($\epsilon$-DP), if for any two databases $x, x' \in \mathcal{X}^n$ such that $d_H(x, x') \leq 1$,
    \begin{align*}
        \Pr[\mathcal{M}(x) \in S] \leq \exp({\epsilon}) \Pr[\mathcal{M}(x') \in S]
    \end{align*}
    for all measurable sets $S$.
\end{definition}

The privacy parameter $\epsilon$, also known as the privacy (loss) budget, captures the worst privacy loss between adjacent data sets. It is used for controlling the trade-off between data privacy and utility. A large $\epsilon$ value may preserve strong data utility but can cause severe privacy leakage. Thus, a typical range for $\epsilon$ is from $0.01$ to $10$ to ensure sufficient privacy protection while maintaining data utility. A common $\epsilon$-DP mechanism is the Laplace mechanism, which adds noise drawn from a Laplace distribution to the summary. As $d_H(x, x') = 1$, it means that the two databases differ in exactly one record, and hence they are called adjacent data sets.

\begin{definition}[Sensitivity]
    let $f: \mathcal{X}^n \to \mathbb{R}$ be a statistic. The sensitivity of $f$ is defined
    \begin{align*}
        \Delta f = \max_{\substack{%
    x, x' \in \mathcal{X}^n\\
     d_H(x, x') \leq 1 \hfill}}
      \!  \| f(x) - f(x') \|_1,
    \end{align*}
    where $\|\cdot\|_1$ is the $\ell_1$-norm.
\end{definition}

The sensitivity of a function quantifies how much the output of the function can change when a single individual's data in the input data set is altered. It is a critical concept in DP as it determines the amount of noise needed to be added to the function's output to achieve a desired level of privacy. 

\begin{proposition} [Laplace Mechanism: \citealp{dwork2006calibrating}]
    Let $M(D) = f(D) + L$, where $L = \text{Lap}\left(0, \frac{\Delta f}{\epsilon} \right)$ with density $f_L(x) = \frac{\epsilon}{2 \Delta f} \exp({-\frac{\epsilon}{\Delta f}|x|})$. Then,  $\mathcal{M}$ satisfies $\epsilon$-DP.    
\end{proposition}

The following are two important properties in $\epsilon$-DP:
\begin{itemize}
    \item (Composition) If a mechanism $\mathcal{M}_1: \mathcal{X}^n \to \mathcal{Y}$ satisfies $\epsilon_1$-DP and another mechanism $\mathcal{M}_2: \mathcal{X}^n \to \mathcal{Z}$ satisfies $\epsilon_2$-DP, then the composed mechanism which jointly releases $(\mathcal{M}_1 (D), \mathcal{M}_2 (D))$ satisfies $(\epsilon_1 + \epsilon_2)$-DP.
    \item (Post-Processing: \citet{dwork2014algorithmic}) If $\mathcal{M}: \mathcal{X}^n \to \mathcal{Y}$ is an $\epsilon$-DP mechanism and $g: \mathcal{Y} \to \mathcal{Z}$ is another mechanism, then $g \circ \mathcal{M}: \mathcal{X}^n \to \mathcal{Z}$ satisfies $\epsilon$-DP.
\end{itemize}

Composition describes the privacy leakage as multiple data summaries are published. Post-processing is a property of differential privacy that ensures the privacy guarantees are preserved under any data-independent transformation of the output, such as our posterior analysis. We use the Laplace mechanisms and their properties in our examples in Section~\ref{sec: simulation}.

\begin{remark}
    We use $\epsilon$-DP for the presentation of this paper primarily for simplicity. Our proposed PF scheme can be easily modified to any other privacy guarantee, especially those with a single privacy parameter, such as $\mu$-GDP \citep{dong_gaussian_2022} or $\rho$-$z$CDP \citep{bun2016concentrated}. Mechanisms for other frameworks, such as $(\epsilon, \delta)$-DP, $f$-DP \citep{dong_gaussian_2022}, or R\'enyi DP \citep{mironov2017renyi}, may need some special treatment to identify a useful filtering schedule in our sampler. 
\end{remark}

\subsection{Particle Filters}

The structure of a particle filtering method can be broken into three main iterative operations: resampling, propagation, and reweighting. Suppose there are $N$ particles in the particle filtering system and $T$ iterations before completion. Iteration $t=0$ is when we initialize the particle system by generating $N$ i.i.d. samples from the prior $\pi_0$. Then, in iteration $t \geq 1$, we perform the following steps subsequently to approximate the intermediate target distribution $\pi_t$, with $\pi_T$ being the ultimate target:

{\emph{Step 1. Resampling.}} 
    The resampling step acts as pruning, helping the system avoid particle degeneration. Although this action introduces variation from the resampling process, it ultimately offsets the risk of particle degeneracy. The simplest resampling method is to perform multinomial resampling, which is simply weighted sampling with replacement:
    \begin{align*}
        \Tilde{y}^{(1,t)},\ldots,\Tilde{y}^{(N,t)} \overset{i.i.d.}{\sim} \text{Multinomial} \left( 1, \left( \mathbf{w}^{(1,t-1)},\ldots, \mathbf{w}^{(N,t-1)} \right) \right),
    \end{align*}
    where $\mathbf{w}^{(j,t-1)}$'s are normalized weights that sum up to 1 and $\tilde{y}^{(j,t)}$'s are resampled particles but not yet perturbed. As $t=1$, the new particles are resampled from the $N$ i.i.d. samples generated from the prior $\pi_0$; otherwise, they are resampled from the independent weighted samples $\{(y, \mathbf{w})^{(j, t-1)}\}_{j=1}^N$ from $\pi_{t-1} (\cdot)$, the intermediate distribution obtained at the end of iteration $t-1$. Other resampling schemes, such as stratified resampling, residual resampling \citep{Liu1998}, and systematic resampling \citep{Carpenter2000}, can introduce less noise and have lower variance in estimates than multinomial resampling.  While we focus on multinomial sampling for simplicity, there should be no complication in extending our methodology to these other resampling schemes, each of which possesses a central limit theorem result \citep{chopin2004}. Additional discussion is provided in Section~\ref{sec: conclusions}.

{\emph{Step 2. Propagation.}}
    The propagation step drives exploration of the parameter space by producing new particles from a perturbation kernel $K_t(\Tilde{y}^{(j, t)}, \cdot)$ for all $j$. A typical particle filtering propagation step ends with new particles approximating
    \begin{align*}
        \eta_t (\cdot) = \int \pi_{t-1} (\tilde{y}) K_t(\tilde{y}, \cdot) d\tilde{y}.
    \end{align*}
    A multivariate Gaussian kernel is often used when targeting a continuous distribution, for it is easily evaluated and sampled. The covariance of a Gaussian kernel can be specified in advance or estimated by the sample covariance from the particles of the previous iteration.  \citet{Beaumont2009} and \citet{filippi2013optimality} derive an optimal kernel variance (scale) with high average acceptance rates for parametric perturbation kernels, such as uniform kernels and Gaussian kernels.

{\emph{Step 3. Reweighting.}} The reweighting step recalibrates the importance of each particle. It uses the proposal distribution $\eta_t(\cdot)$ to target the iteration-wise distribution of interest $\pi_t (\cdot)$ by considering the importance weight
    \begin{align} \label{eq: IS weight}
        w_t(y) = \frac{\pi_t(y)}{\eta_t(y)},
    \end{align}
     which comes from the concept of importance sampling. Importance sampling is a technique used to estimate properties of a particular distribution $\pi$ while only having samples generated from another distribution $\eta$. In such cases, $\ex_{\pi}[\varphi(Y)] = \int \varphi(y) \frac{\pi(y)}{\eta (y)} \eta(y) \, dy$ can be estimated by
    \begin{align*}
        \frac{1}{N} \sum_{i=1}^N \varphi(y_i) \frac{\pi(y_i)}{\eta (y_i)},
    \end{align*}
    where $Y \sim \pi$ but can only sample from a tractable proposal distribution $\eta(y)$. The ratio of the target to the proposal $\pi/\eta$ is called the importance weight.

\begin{remark} \label{remark: importance weight up to a constant}
    The weight calculated in the reweighting step (Line 11) of Algorithm~\ref{alg: DP-PF} is only proportional to the ``ideal'' importance weight, that is, the fraction of the target posterior to the proposal. However, this poses no issue as normalization is performed in the next step (Line 12).
\end{remark}

Algorithm~\ref{alg: particle filters} describes how a particle filtering system works and evolves. Figure~\ref{fig: diagram of particle filters} shows that as iterations proceed, the particles in the system gravitate towards a better proposal $\eta_t$ than the prior $\pi_0$. The colored (overlapping) area expands in $t$ and hence more particles are considered effective. On the other hand, if $\pi_0$ is used to target $\pi$, it reduces to a typical importance sampling (that is, the leftmost graph). In this case, the effective area is small.

\algrenewcommand\algorithmicrequire{\textbf{Input:}}
\algrenewcommand\algorithmicensure{\textbf{Output:}}

\begin{algorithm}[t]
\caption{Generic Particle Filter} \label{alg: particle filters}
\begin{algorithmic}
\Require number of particles $N$, number of iterations $T$, prior $\pi_0$, perturbation kernel $K_t$ \ $(t=1,\ldots,T)$
    \State \textbf{Initialize} by generating, from $\pi_0$, $N$ particles $\{y^{(i,0)}\}_{i=1}^N$ with weights $w_0 = \frac{1}{N}$.
    \For{t = 1,\ldots,T}
        \State Step 1. \textbf{Resample} from $y_{t-1}$ with weights $w_{t-1}$ to get candidate particles \( \{\Tilde{y}^{(i,t)}\}_{i=1}^N \).    
        \State Step 2. \textbf{Propagate} by a perturbation kernel $y^{(i,t)} \sim K_t (\Tilde{y}^{(i, t)}, \cdot)$.
        \State Step 3. \textbf{Reweight} $y^{(i,t)}$ with importance weights as in \eqref{eq: IS weight}.
    \EndFor
\Ensure $\{ (y, \mathbf{w})^{(i, T)}\}_{i=1}^N$: weighted samples from $\pi_T ( \cdot )$.
\end{algorithmic}
\end{algorithm}

\begin{figure}[t]
    \begin{tikzpicture}[scale=1.12]
    \node[thin, black] at (0.6,0.9) {$\pi_0$};
    \node[thin, cyan!50!black] at (1.8,2.7) {$\pi$};
    \begin{axis}[
      no markers, domain=0:10, samples=100,
      axis lines*=left, xlabel=$y$,
      every axis y label/.style={at=(current axis.above origin),anchor=south},
      every axis x label/.style={at=(current axis.right of origin),anchor=west},
      height=4cm, width=4cm,
      xtick=\empty, ytick=\empty,
      enlargelimits=false, clip=false, axis on top,
      grid = major
      ]
      \addplot [very thick, black] {gauss(2,2)};
      \addplot [fill=red!20, draw=none, domain=5.75:10] {gauss(2,2)} \closedcycle;
      \addplot [densely dotted, very thick, cyan!50!black] {gauss(7,0.5)};
      \addplot [fill=red!20, draw=none, domain=0:5.75] {gauss(7,0.5)} \closedcycle;
    \end{axis}
\end{tikzpicture}
\begin{tikzpicture}[scale=1.12]
    \node [] at (0.7,1.2) {$\eta_1$};
    \node[thin, cyan!50!black] at (1.8,2.7) {$\pi$};
    \begin{axis}[
      no markers, domain=0:10, samples=100,
      axis lines*=left, xlabel=$y$,
      every axis y label/.style={at=(current axis.above origin),anchor=south},
      every axis x label/.style={at=(current axis.right of origin),anchor=west},
      height=4cm, width=4cm,
      xtick=\empty, ytick=\empty,
      enlargelimits=false, clip=false, axis on top,
      grid = major
      ]
      \addplot [dotted, thick, black] {gauss(2,2)};
      \addplot [very thick, black] {gauss(3,1.5)};
      \addplot [fill=red!20, draw=none, domain=5.75:10] {gauss(3,1.5)} \closedcycle;
      \addplot [densely dotted, very thick, cyan!50!black] {gauss(7,0.5)};
      \addplot [fill=red!20, draw=none, domain=0:5.75] {gauss(7,0.5)} \closedcycle;
    \end{axis}
\end{tikzpicture}
\begin{tikzpicture}[scale=1.12]
    \node [] at (0.9,1.5) {$\eta_2$};
    \node[thin, cyan!50!black] at (1.8,2.7) {$\pi$};
    \node [] at (2.7,1.2) {$\cdot\cdot\cdot$};
    \begin{axis}[
      no markers, domain=0:10, samples=100,
      axis lines*=left, xlabel=$y$,
      every axis y label/.style={at=(current axis.above origin),anchor=south},
      every axis x label/.style={at=(current axis.right of origin),anchor=west},
      height=4cm, width=4cm,
      xtick=\empty, ytick=\empty,
      enlargelimits=false, clip=false, axis on top,
      grid = major
      ]
      \addplot [dotted, thick, black] {gauss(2,2)};
      \addplot [very thick, black] {gauss(4,1)};
      \addplot [fill=red!20, draw=none, domain=5.75:10] {gauss(4,1)} \closedcycle;
      \addplot [densely dotted, very thick, cyan!50!black] {gauss(7,0.5)};
      \addplot [fill=red!20, draw=none, domain=0:5.75] {gauss(7,0.5)} \closedcycle;
    \end{axis}
\end{tikzpicture}
\begin{tikzpicture}[scale=1.12]
    \node [] at (1.2,1.8) {$\eta_{T}$};
    \node[thin, cyan!50!black] at (1.8,2.7) {$\pi$};
    \begin{axis}[
      no markers, domain=0:10, samples=100,
      axis lines*=left, xlabel=$y$,
      every axis y label/.style={at=(current axis.above origin),anchor=south},
      every axis x label/.style={at=(current axis.right of origin),anchor=west},
      height=4cm, width=4cm,
      xtick=\empty, ytick=\empty,
      enlargelimits=false, clip=false, axis on top,
      grid = major
      ]
      \addplot [dotted, thick, black] {gauss(2,2)};
      \addplot [very thick, black] {gauss(5.5,0.75)};
      \addplot [fill=red!20, draw=none, domain=6.25:10] {gauss(5.5,0.75)} \closedcycle;
      \addplot [very thick, cyan!50!black] {gauss(7,0.5)};
      \addplot [fill=red!20, draw=none, domain=0:6.25] {gauss(7,0.5)} \closedcycle;
    \end{axis}
\end{tikzpicture}
\caption{Diagram of a particle filter. Plots progress from left to right.}
\label{fig: diagram of particle filters}
\end{figure}

\begin{example}[ABC SMC: \citealp{toni_approximate_2009}]

For Bayesian analysis, PF is commonly used in combination with ABC, which is called ABC SMC. ABC SMC, described in Algorithm~\ref{alg: abc smc}, incorporates an additional rejection step to rule out bad particles when its threshold $d(s,\cdot) > \Delta$. The function $d(s,\cdot)$ measures the distance between a statistic and its simulated counterpart, accepting particles within a range of $\Delta$. Smaller thresholds require more computation to generate effective particles but can result in a smaller approximation error.

By considering the events of $\{d(s,\cdot)\leq \Delta_t\}$ over iterations, ABC SMC can adaptively approach its target by shrinking $\Delta_t$ in $t$. This is essentially the ``schedule'' in ABC SMC. More formally, the schedule targets $\pi(\theta \mid d(s,\cdot)\leq \Delta_t)$, where the tolerance $\Delta_t$ decreases with $t$ until reaching the target ABC posterior distribution at the final step $T$. This process avoids rejecting too many particles in the first few iterations when the particle system is still close to prior knowledge. 
\end{example}

\begin{algorithm}[t]
\caption{ABC SMC} \label{alg: abc smc}
\begin{algorithmic}
    \Require number of particles $N$, number of iterations $T$, prior $\pi_0$, perturbation kernel $K_t$, schedule $\Delta_t$ \ $(t=1,\ldots,T)$
        \State \textbf{Initialize} by generating, from $\pi_0$, $N$ particles $\{y^{(i,0)}\}_{i=1}^N$ with weights $w_0 = \frac{1}{N}$
        \For{t = 1,\ldots,T}
            \State Step 1. \textbf{Resample} from $y_{t-1}$ with weights $w_{t-1}$ to get candidate particles \( \{\Tilde{y}^{(i,t)}\}_{i=1}^N \).    
            \State Step 2. \textbf{Propagate} by a perturbation kernel $y^{(i,t)} \sim K_t (\Tilde{y}^{(i, t)}, \cdot)$.
            \State Step 2.5. \textbf{Reject} when $d > \Delta_t$ and go back to Step 2.
            \State Step 3. \textbf{Reweight} $y^{(i,t)}$ with importance weights as in \eqref{eq: IS weight}.
        \EndFor
    \Ensure $\{ (y, \mathbf{w})^{(i, T)}\}_{i=1}^N$: weighted samples from $\pi_T ( \cdot )$.
\end{algorithmic}
\end{algorithm}


\section{Particle Filter for Differential Privacy} \label{sec: particle filters for DP}

In this section, we present our particle filtering algorithm for DP, which generates weighted samples that approximately follow the private posterior distribution $\pi_\epsilon(\theta, x \mid \sdp)$. 

Our algorithm integrates the strengths of ABC SMC and DP-Reject-ABC. While ABC SMC effectively narrows down potential parameters, it introduces approximation errors as its final target distribution is an ABC posterior rather than the true posterior distribution. To avoid this issue, we incorporate techniques from DP-Reject-ABC, which produces i.i.d. samples from the exact posterior. To properly combine the two methods, we define the $\epsilon_t$ schedule as an increasing, positive sequence for the $\epsilon$ budget, with $\epsilon_T = \epsilon$. Additionally, the scale of the perturbation kernels can be set as a decreasing function of $t$, narrowing the exploration scope as the iterations proceed. The $\epsilon_t$-schedule is analogous to the decreasing sequence $\Delta_t$ in ABC SMC, which fits the empirical evidence that smaller $\epsilon$ values result in a posterior that is closer to the prior \citep{ju2022data}.

We point out the key differences of our differentially private particle filtering algorithm compared to the previous two PF algorithms. Compared to Algorithm~\ref{alg: abc smc}, \emph{Steps 1. and 2.} are the same except that our algorithm requires specifying the $\epsilon_t$-schedule.
     
{\emph{Step 2.5. Rejection.}} The DP propagation step incorporates a rejection step (in Lines 7 to 9 in Algorithm~\ref{alg: DP-PF}) that helps target the exact intermediate private posterior $\pi_{\epsilon_t}(\cdot \mid \sdp)$. Instead of function $d$ in ABC, we have the probability of acceptance
\begin{align*}
 r_t(x) = \frac{m_{\epsilon_t}(\sdp \mid x)}{\sup_{x} m_{\epsilon_t}(\sdp \mid x)},   
\end{align*}
where $x$ is sampled from $f \left( \cdot \mid {\theta} \right)$. This is a direct and crucial consequence of the design of Algorithm~\ref{alg: DP-PF}, in which the mechanism density appears only in the propagation step but not in the reweighting or resampling steps.  The tempering effect on the mechanism density $m_{\epsilon_t}$ enables the particle system to collect plausible particles in early iterations to avoid a high rejection rate and filter these through later iterations to meet the actual $\epsilon$-DP guarantee. 
      
{\emph{Step 3. Reweighting.}} As usual, the reweighting step considers the importance weight. The numerator now becomes our target private posterior $\pi_{\epsilon_t} (\theta, x \mid \sdp)$ and the denominator is the whole propagation density $h_t(\theta, x)$, which involves both perturbation and rejection. The properties of $h_t$ are discussed in \Cref{appendix: proofs}.

\subsection{Algorithm} \label{subsec: algorithm}

In this section, we present the proposed algorithm, detailing its inputs, outputs, and assumptions, and discuss how its design underpins the theoretical guarantees established later.

\algrenewcommand\algorithmicrequire{\textbf{Input:}}
\algrenewcommand\algorithmicensure{\textbf{Output:}}

\begin{algorithm}[t]
\caption{DP-PF} \label{alg: DP-PF}
\begin{algorithmic}[1]
\Require privacy budget $\epsilon$, sample size $n$, particles $N$, iterations $T$, prior $\pi_0$, privatized summary $\sdp$, sampler $f$, kernel $K_t$, schedule $\epsilon_t$, mechanism densities $m_{\epsilon_t}$ \ ($t=1,\ldots,T$)
    \For{$t = 0,1,\ldots,T$}
        \State Set $\epsilon_t$ and variance (scale) of $K_t$ when $t>0$
        \For{$i = 1,..,N$}
            \State \textbf{[Initialization]} Sample $\theta^{(i,t)} \sim \pi_0$ if $t = 0$; else:
            \State \quad \textbf{[Resampling]} Sample $\Tilde{\theta} \sim \{\theta^{(j, t-1)}\}_{j=1}^N$ using weights $\{\mathbf{w}^{(j, t-1)}\}_{j=1}^N$.
            \State \quad \textbf{[Propagation]} Sample $\theta \sim K_t(\Tilde{\theta}, \cdot)$ until $\pi_0 (\theta) > 0$.
            \State \quad Sample $x \sim f(\cdot \mid \theta)$, compute 
            \[
            r_t(x) = \frac{m_{\epsilon_t}(\sdp \mid x)}{\sup_{x} m_{\epsilon_t}(\sdp \mid x)}.
            \]
            \State \quad Sample $u \sim U(0,1)$. 
            \State \quad If $u > r_t(x)$, go to the propagation step; else: set $\theta^{(i,t)} \gets \theta$ and $x^{(i,t)} \gets x$.
        \EndFor
        \State \textbf{[Reweighting]} Compute weights for all $i=1,\ldots,N$:
        \[
        \tilde{w}^{(i,t)} = 
        \begin{cases} 
            1, & t=0, \\
            \tilde{w}^{(i,t)} \gets \frac{\pi_0 \left( \theta^{(i,t)} \right)}{\sum_{j=1}^N \mathbf{w}^{(j, t-1)} K_t(\theta^{(j,t-1)}, \theta^{(i,t)} \mid \pi_0 (\theta^{(i,t)}) > 0, r_t(x^{(i,t)} > U)}, & \text{otherwise}.
        \end{cases}
        \]
        \State Normalize the weights to get $\{\mathbf{w}^{(j, t)}\}$ s.t. $\sum_{j=1}^N \mathbf{w}^{(j, t)} = 1$.
    \EndFor
\Ensure $\{ (\theta, x, \mathbf{w})^{(j, T)}\}_{j = 1}^N$: weighted samples from $\pi_{\epsilon} \left( \theta, x \mid \sdp \right)$.
\end{algorithmic}
\end{algorithm}

For the input, we assume the following:
\begin{enumerate}
    \item[A.1] prior $\pi_0$ can be sampled and its density can be evaluated,
    \item[A.2] data model $f$ can be sampled,
    \item[A.3] perturbation kernel $K_t$ can be sampled, and its density can be evaluated for each $t$,
    \item[A.4] mechanism density $m_{\epsilon_t}(\sdp|\cdot)$ can be evaluated and has a maximum for each $t$,
    \item[A.5] $\ex_{\pi_{\epsilon_T} (\theta, x | \sdp)} [\varphi(\theta, x)]$, $c_{\eta_t} (\sdp)$ and $c_{\epsilon_t}(\sdp)$ are all finite.
\end{enumerate}

For the output, $\{(\theta, x, \mathbf{w})^{(j, T)}\}_{j = 1}^N$ are weighted samples that are approximately distributed as  $\pi_{\epsilon} \left( \theta, x \mid \sdp \right)$. More importantly, by taking the weighted average of the particles, we obtain the consistent estimator of $\ex_{\pi_{\epsilon_T} (\theta, x | \sdp)} [\varphi(\theta, x)]$ for a given function $\varphi$ to be
\begin{align} \label{eq: the consistent estimator}
    \hat{\ex}_T [\varphi] \coloneqq \sum_{i=1}^N \textbf{w}^{(i,T)} \varphi \left( (\theta, x)^{(i,T)} \right).
\end{align}
The properties of \eqref{eq: the consistent estimator} are studied in Section~\ref{subsec: consistency}, where we also state additional technical conditions on the prior \( \pi_0 \) and on the perturbation kernel \( K_t \) to ensure that the weight function \( w_t \) is well defined, with \( 0 \leq w_t(\theta) < \infty \) for all \( \theta \in \Theta_t \) at each \( t \).

Algorithm~\ref{alg: DP-PF} presents our differentially private particle filtering (DP-PF) approach for posterior inference. The for loop on Line 3 and the unnormalized weighting calculations on Line 11 can be processed in parallel across particles, substantially reducing computation time.

\begin{remark} 
    To decide what $T$ to use, we note the following intuition: when rejection sampling maintains good acceptance and efficiency (for example, when  $\epsilon$ and sample size $n$ are small), we let $T$ be small (close to $1$); conversely, we let $T$ be large ($\geq 10$).
\end{remark}

While many variations are possible, there is an inherent trade-off between computational cost and particle quality, especially in high dimensions. For instance, a DP importance sampling variant, relying only on perturbation and reweighting, can start with hundreds of thousands of particles efficiently, but without rejection, the weights concentrate quickly, leading to severe degeneracy and very small ESS. Conversely, incorporating ideas from \citet{ju2022data}, which sequentially updates and stores synthetic data, may yield more stable but potentially poorly mixed particles at substantially greater computational cost.

DP-PF is designed not only to strike a balance between these two extremes, but also to exploit an advantage that alternative designs do not possess. The key structural feature is that the data and mechanism densities cancel in the weight update (Algorithm~3, Line~11), yielding a particularly simple form. As a result, in the reweighting step, the data model $f$ need not be evaluated, and neither the model density nor the mechanism density appears.

This structure has two important consequences. First, the acceptance probability depends only on the proposal $\eta_t$ (or kernel $K_t$); see Proposition~\ref{prop: acceptance distribution}. Second, the marginal likelihood admits a clean characterization that the incremental evidence ratios depend only on the mechanism density; see Proposition~\ref{thm: evidence}. Moreover, since the mechanism density enters only through the rejection step and not the reweighting step, the algorithm integrates naturally into the classical SMC framework: the propagation step absorbs the rejection, while the weighting step remains unaffected. This cancellation is particularly valuable in high dimensions, as it avoids repeated evaluation of expensive data and DP mechanisms.

\begin{figure}[t]
    \begin{tikzpicture}[scale=1.12]
        \node[thin, black] at (0.5,0.9) {$\pi$};
        \node[thin, cyan!50!black] at (2,2.8) {$\pi_{\epsilon} (\cdot \mid \sdp)$};
        \begin{axis}[
          no markers, domain=0:10, samples=100,
          axis lines*=left, xlabel=$\theta$,
          every axis y label/.style={at=(current axis.above origin),anchor=south},
          every axis x label/.style={at=(current axis.right of origin),anchor=west},
          height=4cm, width=4cm,
          xtick=\empty, ytick=\empty,
          enlargelimits=false, clip=false, axis on top,
          grid = major
          ]
          \addplot [very thick, black] {gauss(2,2)};
          \addplot [fill=red!20, draw=none, domain=5.75:10] {gauss(2,2)} \closedcycle;
          \addplot [densely dotted, very thick, cyan!50!black] {gauss(7,0.5)};
          \addplot [fill=red!20, draw=none, domain=0:5.75] {gauss(7,0.5)} \closedcycle;
        \end{axis}
    \end{tikzpicture}
    \begin{tikzpicture}[scale=1.12]
        \node [] at (0.6,1.2) {$h_1$};
        \node[thin, cyan!50!black] at (1,1.5) {$\pi_{\epsilon_1}$};
        \begin{axis}[
          no markers, domain=0:10, samples=100,
          axis lines*=left, xlabel=$\theta$,
          every axis y label/.style={at=(current axis.above origin),anchor=south},
          every axis x label/.style={at=(current axis.right of origin),anchor=west},
          height=4cm, width=4cm,
          xtick=\empty, ytick=\empty,
          enlargelimits=false, clip=false, axis on top,
          grid = major
          ]
          \addplot [dotted, thick, black] {gauss(2,2)};
          \addplot [very thick, black] {gauss(3,1.5)};
          \addplot [thick, cyan!50!black] {gauss(4,1)};
          \addplot [fill=red!20, draw=none, domain=3:10] {gauss(3,1.5)} \closedcycle;
          \addplot [densely dotted, very thick, cyan!50!black] {gauss(7,0.5)};
          \addplot [fill=red!20, draw=none, domain=0:3] {gauss(4,1)} \closedcycle;
        \end{axis}
    \end{tikzpicture}
    \begin{tikzpicture}[scale=1.12]
        \node [] at (0.9,1.5) {$h_2$};
        \node[thin, cyan!50!black] at (1.3,1.9) {$\pi_{\epsilon_2}$};
        \node [] at (2.7,1.2) {$\cdot\cdot\cdot$};
        \begin{axis}[
          no markers, domain=0:10, samples=100,
          axis lines*=left, xlabel=$\theta$,
          every axis y label/.style={at=(current axis.above origin),anchor=south},
          every axis x label/.style={at=(current axis.right of origin),anchor=west},
          height=4cm, width=4cm,
          xtick=\empty, ytick=\empty,
          enlargelimits=false, clip=false, axis on top,
          grid = major
          ]
          \addplot [dotted, thick, black] {gauss(2,2)};
          \addplot [very thick, black] {gauss(4,1)};
          \addplot [thick, cyan!50!black] {gauss(5.5,0.75)};
          \addplot [fill=red!20, draw=none, domain=4.75:10] {gauss(4,1)} \closedcycle;
          \addplot [densely dotted, very thick, cyan!50!black] {gauss(7,0.5)};
          \addplot [fill=red!20, draw=none, domain=0:4.75] {gauss(5.5,0.75)} \closedcycle;
        \end{axis}
    \end{tikzpicture}
    \begin{tikzpicture}[scale=1.12]
        \node [] at (1.3,1.9) {$h_T$};
        \node[thin, cyan!50!black] at (2.1,2.8) {$\pi_{\epsilon} (\cdot \mid \sdp)$};
        \begin{axis}[
          no markers, domain=0:10, samples=100,
          axis lines*=left, xlabel=$\theta$,
          every axis y label/.style={at=(current axis.above origin),anchor=south},
          every axis x label/.style={at=(current axis.right of origin),anchor=west},
          height=4cm, width=4cm,
          xtick=\empty, ytick=\empty,
          enlargelimits=false, clip=false, axis on top,
          grid = major
          ]
          \addplot [dotted, thick, black] {gauss(2,2)};
          \addplot [very thick, black] {gauss(5.5,0.75)};
          \addplot [fill=red!20, draw=none, domain=6.25:10] {gauss(5.5,0.75)} \closedcycle;
          \addplot [very thick, cyan!50!black] {gauss(7,0.5)};
          \addplot [fill=red!20, draw=none, domain=0:6.25] {gauss(7,0.5)} \closedcycle;
        \end{axis}
    \end{tikzpicture}
    \caption{Diagram of DP-PF. Plots progress from left to right.}
    \label{fig: diagram of DP particle filters}
\end{figure}

Unlike Figure~\ref{fig: diagram of particle filters}, the colored area in Figure~\ref{fig: diagram of DP particle filters} increases as a consequence of two features of our proposed design: (i) the particle filtering procedure progressively shifts the proposals toward the private posterior, and (ii) the intermediate target distributions are flatter than the private posterior, thereby making the sampling more efficient at each step.

\subsection{Consistency} \label{subsec: consistency}

In this section, we prove that our estimator \eqref{eq: the consistent estimator} is consistent as the number of particles $N$ tends to infinity.

Since our goal is to sample from and perform inference based on $\pi_{\epsilon_T} (\theta, x \mid \sdp)$ (or simply its marginal $\pi_{\epsilon_T} (\theta \mid \sdp)$), we assess the convergence of the weighted estimator \eqref{eq: the consistent estimator} to its posterior expectation, componentwise:
    \begin{align*} 
    \ex_T [\varphi] \coloneqq \ex_{\pi_{\epsilon_T} (\theta,x \mid \sdp)} [\varphi(\theta,x)]
    \end{align*}
for any measurable $\varphi$ mapping from the product of the parameter space and the database space to $d$-dimensional real space such that $\ex_T[\varphi]$ is finite. Throughout Section~\ref{sec: particle filters for DP}, moment conditions involving vector-valued $\varphi$ are understood componentwise.

\begin{restatable}[Strong Consistency]{thm}{strongConsistency} \label{thm: posterior mean convergence}
Let $\varphi: \Theta_t \times \mathcal{X}^n \to \mathbb{R}^d$, such that $\ex_{h_t} |w_t \varphi_j| < \infty$ for all $j=1,\ldots,d$. Our estimator, defined in \eqref{eq: the consistent estimator}, converges almost surely:
     \begin{align*}
          \hat{\ex}_t [\varphi] = \sum_{i=1}^N \textup{\textbf{w}}^{(i,t)} \varphi \left( (\theta, x)^{(i,t)} \right) \overset{a.s.}{\to} \ex_t [\varphi]
     \end{align*}
as $N \to \infty$. In particular, for $t=T$, $\hat{\ex}_T[\varphi] \overset{a.s.}{\to} \ex_T[\varphi]$.
\end{restatable}

Theorem~\ref{thm: posterior mean convergence} plays a role analogous to ergodicity in DP-DA-MCMC. The almost sure convergence is with respect to the randomness of the PF while the DP summary is held fixed. Although the DP-PF importance weight $\tilde{w}_t$ is only proportional to the ideal weight
$$
w_t=\frac{\pi_{\epsilon_t}(\theta,x\mid\sdp)}{h_t(\theta,x)},
$$
the unknown proportionality constant is not an issue: it can be consistently estimated by setting $\varphi=1$. Consequently, both the numerator and denominator of the self-normalized importance sampling estimator are consistent, and $\hat{\ex}_t[\varphi]$ converges almost surely to $\ex_t[\varphi]$ by the continuous mapping theorem. Theorem~\ref{thm: posterior mean convergence} therefore provides consistent estimates of posterior moments, such as $\varphi(\theta)=\theta^q$ for $q\in\mathbb{N}$, as well as posterior probabilities, obtained by taking $\varphi(\theta)=1_{\{\theta\in S\}}$ for a measurable set $S\in\Theta_t$. While our primary interest is in functions of $\theta$, functions of $x$ may also be considered.

\begin{remark} \label{remark: consistency assumption}
    $\ex_{h_t} |w_t \varphi| < \infty$ is indeed a crucial but mild assumption in the following way. 
    \begin{itemize}    
        \item It ensures that $\ex_t [\varphi]$ exists.
        \item It excludes the cases in which $w_t = \infty$ for any $\theta$.
        \item It is mild because it is often reasonable to assume that the parameter space $\Theta_t$ and the weight $w_t$ will be bounded in practice. In Euclidean space, if $\Theta_t$ is a compact set, for any continuous function $\varphi$ defined on $\Theta_t$,
    \begin{align*}
        \ex_{h_t} |w_t \varphi| \leq \sup_{\theta \in \Theta_t} |\varphi(\theta)| < \infty.
    \end{align*}
    \end{itemize}
\end{remark}

\subsection{Central Limit Theorem} \label{subsec: clt}

In addition to the consistency of $\hat{\ex}_T[\varphi]$, it is also crucial to understand its stability. In this section, we provide the asymptotic distribution of the consistent estimator $\hat{\ex}_T[\varphi]$, its asymptotic variance, and a $1-\alpha$ confidence interval. We follow a standard central limit theorem (CLT) argument for particle filtering algorithms by \citet{chopin2004}, but need to adjust for the parts that involve our DP components and rejection step; specifically, the propagation step. As in the previous section, this quantifies the randomness introduced by the PF system, not the randomness from the privacy mechanism or the data.

To make the CLT hold for the propagation step, we need to assume that
\begin{align*}
    \ex_{h_t} | w_t \varphi |^{2+\zeta} < \infty
\end{align*}
for some $\zeta > 0$. It is worth noting that we only need to assume $\ex_{h_t}|w_t \varphi|<\infty$ for consistency, but an additional finite absolute $2+\zeta$ moment is required for the CLT. However, this is still fairly mild in a similar sense as discussed in Remark~\ref{remark: consistency assumption}. In particular, the boundedness of $\Theta_t$ and the weights $w_t$ implies that this condition holds.

\begin{remark}
    The assumption $\ex_{h_t} |w_t \varphi|^{2+\zeta} < \infty$ with some $\zeta>0$ ensures Lyapunov's condition for the CLT. It also implies $\ex_{h_t} |w_t \varphi| < \infty$ by Jensen's inequality, and hence guarantees strong consistency.
\end{remark}

\begin{restatable}[Central Limit Theorem]{thm}{clt} \label{thm: clt}
    Let $\varphi: \Theta_t \times \mathcal{X}^n \to \mathbb{R}^d$ satisfy the
    recursive measurability and moment conditions stated in
    \Cref{appendix: proofs}, and suppose there exists $\zeta > 0$ such that $\ex_{h_t} |w_t \varphi_j|^{2+\zeta} < \infty$ for all $j=1,\ldots,d$. Then,
    \begin{align*}
        \sqrt{N} \left( \hat{\ex}_t[\varphi] - \ex_t[\varphi] \right) \overset{d}{\to} N(0, V_t(\varphi))
    \end{align*}
    as $N \to \infty$.
\end{restatable}

Theorem~\ref{thm: clt} plays a role analogous to geometric ergodicity in DP-DA-MCMC. However, a relatively strong assumption has to be made for geometric ergodicity to hold for DP-DA-MCMC: there exists $0 < l \leq u < \infty$ such that $l\leq f(x \mid \theta)\leq u$ for all $\theta,x$ \citep{ju2022data}. In contrast, Theorem \ref{thm: clt} has a relatively mild assumption. In addition, DP-DA-MCMC also requires $\epsilon$-DP, while DP-PF allows for other privacy frameworks, such as $\mu$-GDP and $\rho$-$z$CDP.

\subsection{Monte Carlo Standard Errors and Asymptotic Confidence Intervals} \label{subsec: mcse and asymptotic ci}

The asymptotic normality established in the previous section characterizes the Monte Carlo error at the $\sqrt{N}$ rate. While this provides a basis for quantifying Monte Carlo standard errors (MCSE), we need to develop an estimator of the MCSE in order to construct confidence intervals (CI) for the true posterior quantities. To this end, we use the concept of effective sample size to account for the weighting of the particles.

\begin{definition}[Effective Sample Size: \citealp{kong1992}]
Suppose the target distribution is $\pi(x)$, the proposal distribution is $h(x)$, and the importance weight is $w(x) = \frac{\pi(x)}{h(x)}$. The quantity
    \begin{align*}
        ESS_N := \frac{N}{1+\var_{h} (w(X))}
    \end{align*}
is defined as the population effective sample size with respect to the distributions $\pi$ and $h$.
\end{definition}

Effective Sample Size (ESS) is a rough measure to evaluate the quality of importance weights in importance sampling. It indicates that $N$ weighted samples from $h$ create the same variation as $ESS_N$ true samples from $\pi$ do. A similar concept, relative numerical efficiency (RNE), is introduced by \citet{Geweke1989}, where $RNE \approx ESS/N$ but can occasionally exceed $1$. Thus, for ordinary importance sampling, the numerical standard error of the estimate is the fraction of $(ESS)^{-1/2}$ or $(RNE \cdot N)^{-1/2}$ of the posterior standard deviation.

\begin{restatable}[ESS estimate]{lma}{essEstimate} \label{lemma: ess estimate}
Let $w_i=\frac{\pi(X_i)}{h(X_i)}$ and $\widehat{ESS}_N=\frac{(\sum_{i=1}^N w_i)^2}{\sum_{i=1}^N w_i^2}$. 
Assume $(X_i)_{i\ge1}$ has marginal distribution $h$ and satisfies $\frac{1}{N}\sum_{i=1}^N w_i \xrightarrow{p} \ex_h[w(X)]$ and $\frac{1}{N}\sum_{i=1}^N w_i^2 \xrightarrow{p} \ex_h[w(X)^2]$, with $\ex_h[w(X)^2]<\infty$. 
Then, as $N\to\infty$,
$$
\widehat{ESS}_N/N \xrightarrow{p} \frac{1}{1+\var_h(w(X))}.
$$
\end{restatable}

The moment assumptions require only a weak law of large numbers for $w(X)$ and $w(X)^2$, rather than independence of the samples. They hold for importance sampling and, more generally, for ergodic Markov chains in MCMC, where empirical averages converge to expectations under the stationary distribution. In particle filtering methods, although resampling induces dependence, standard consistency results ensure convergence of empirical averages under mild regularity conditions. Thus, independence is sufficient but not necessary for the consistency of $\widehat{ESS}_N$. Moreover, the definition of $\widehat{ESS}_N$ is invariant to normalization. Combined with Lemma~\ref{lemma: ess estimate}, this provides a basis for quantifying the MCSE and constructing an asymptotic confidence interval.

Theorem~\ref{thm: CI} establishes a consistent estimator of the asymptotic variance $V_t(\varphi)$ and the resulting asymptotic confidence interval for $\ex_t[\varphi]$. Unlike the standard approximation that omits the remainder term, our variance estimator retains this term to account for its potential contribution.

\begin{restatable}[Asymptotic Confidence Interval]{thm}{CI} \label{thm: CI}
    Let $\varphi: \Theta_t \times \mathcal{X}^n \to \mathbb{R}$ satisfy the
    recursive measurability and moment conditions stated in
    \Cref{appendix: proofs}, and suppose there exists $\zeta > 0$ such that $\ex_{h_t} |w_t \varphi|^{2+\zeta} < \infty$. Then,
    \begin{enumerate}
        \item a consistent estimate of $V_t(\varphi)$ is
        \begin{align*}
            \hat{V_t}(\varphi) &\coloneqq \frac{\sum_{i=1}^N \textup{\textbf{w}}^{(i,t)} \left[ \varphi \left( (\theta, x)^{(i,t)} \right) - \hat{\ex}_t [\varphi] \right]^2}{\widehat{ESS_N}/N} \\
            &\quad+ \sum_{i=1}^N \textup{\textbf{w}}^{(i,t)} \left[\textup{\textbf{w}}^{(i,t)} -\frac{1}{N} \right] \left[ \varphi \left( (\theta, x)^{(i,t)} \right) - \hat{\ex}_t [\varphi] \right]^2;
        \end{align*}
        \item an asymptotic $(1-\alpha)$ confidence interval for $\ex_t[\varphi]$ is
    \begin{align*}
        \hat{\ex}_t[\varphi] \pm z_{1-\frac{\alpha}{2}}\sqrt{\frac{\hat{V}_t(\varphi)}{N}},
    \end{align*}
    where $z_{1-\frac{\alpha}{2}}$ is the $\left(1-\frac{\alpha}{2}\right)$ quantile of a standard normal distribution.
    \end{enumerate}
\end{restatable}

\begin{remark}
In much of the literature, the second term of $V_t (\varphi)$, also known as the remainder term, is omitted under the assumption that it is negligible. However, \citet{kong1992} emphasize that the approximation may become inaccurate if the remainder term is not sufficiently small. Since we can estimate the remainder term by the second term in $\hat{V}_t (\varphi)$, we retain it in our estimator in Theorem~\ref{thm: CI} to address this potential issue.
\end{remark}

\subsection{Marginal Likelihood of Privatized Data} \label{subsec: evidence}

In classical SMC methods for Bayesian inference, particles represent parameter values and are reweighted according to the likelihood of the observed data. This sequential reweighting implicitly accumulates the normalizing constants of the intermediate distributions, and their product yields the marginal likelihood of the observed data \citep{del2006sequential}. Since our inference is based on the privatized statistic $\sdp$, the corresponding normalizing constants define its marginal likelihood, providing a differentially private counterpart to the classical model evidence with a similar telescoping structure.

A key structural property of Algorithm~\ref{alg: DP-PF} makes this construction clean. As noted in Section~\ref{subsec: algorithm}, the data and mechanism densities cancel in the weight update (Line~11), so $m_{\epsilon_t}$ enters the algorithm only through the propagation step (Line~7), not through reweighting. Consequently, ratios of mechanism densities can be evaluated as importance-sampling expectations under the intermediate private posteriors using the particle approximation already maintained by Algorithm~\ref{alg: DP-PF}, with no additional simulation.

For notational clarity, let
\begin{align*}
    \kappa_t(\theta, x) \;=\; \pi_0(\theta)\, f(x \mid \theta)\, m_{\epsilon_t}(\sdp \mid x), \qquad
    c_{\epsilon_t}(\sdp) \;=\; \int_{\Theta} \int_{\mathcal{X}^n} \kappa_t(\theta, x) \, dx \, d\theta,
\end{align*}
so that $\pi_{\epsilon_t}(\theta, x \mid \sdp) = \kappa_t(\theta, x) / c_{\epsilon_t}(\sdp)$. We adopt the convention $m_{\epsilon_0}(\sdp \mid x) \equiv 1$ and $c_{\epsilon_0}(\sdp) = 1$ as a result, corresponding to the degenerate ``no-release'' baseline ($\epsilon_0 = 0$). Under this convention, $\pi_{\epsilon_0}(\theta, x \mid \sdp) = \pi_0(\theta) f(x \mid \theta)$ coincides with the prior predictive distribution, which is exactly the distribution sampled at the initialization step (Line~4) of Algorithm~\ref{alg: DP-PF}. We further define the incremental evidence ratio
\begin{align*}
    \rho_t \;\coloneqq\; \frac{c_{\epsilon_t}(\sdp)}{c_{\epsilon_{t-1}}(\sdp)}, \qquad t = 1, \ldots, T,
\end{align*}
so that the marginal likelihood factorizes as $c_{\epsilon_T}(\sdp) = \prod_{t=1}^T \rho_t$.

\begin{restatable}[Marginal likelihood and its consistent estimation]{thm}{evidence} \label{thm: evidence}
Let $\{\epsilon_t\}_{t=0}^T$ be the schedule of Algorithm~\ref{alg: DP-PF} with $\epsilon_0 = 0$. For $t = 1, \ldots, T$, $\ex_{h_t}\!\left[\, w_t{m_{\epsilon_t}(\sdp \mid x)}/{m_{\epsilon_{t-1}}(\sdp \mid x)} \,\right] \;<\; \infty$. Then,
\begin{enumerate}
    \item Each incremental evidence ratio admits the representation
    \begin{align*}
        \rho_t
        \;=\;
        \ex_{\pi_{\epsilon_{t-1}}(\theta, x \mid \sdp)}\!\left[\frac{m_{\epsilon_t}(\sdp \mid x)}{m_{\epsilon_{t-1}}(\sdp \mid x)}\right],
    \end{align*}
    so that $c_{\epsilon_T}(\sdp) = \prod_{t=1}^T \rho_t$.
    \item Each $\rho_t$ is consistently estimated, as $N \to \infty$, by
    \begin{align*}
        \hat{\rho}_t
        \;=\;
        \sum_{i=1}^{N} \textup{\textbf{w}}^{(i,t-1)} \,
        \frac{m_{\epsilon_t}(\sdp \mid x^{(i,t-1)})}{m_{\epsilon_{t-1}}(\sdp \mid x^{(i,t-1)})},
    \end{align*}
    where $\textup{\textbf{w}}^{(i,t)}$ and $x^{(i,t)}$ are the normalized weights and accepted synthetic data set produced at iteration~$t$ of Algorithm~\ref{alg: DP-PF}. Consequently, $\widehat{c_{\epsilon_T}(\sdp)} = \prod_{t=1}^T \hat{\rho}_t$ is a consistent estimator of $c_{\epsilon_T}(\sdp)$.
\end{enumerate}
\end{restatable}

Theorem~\ref{thm: evidence} is the differentially private analogue of the marginal-likelihood identity for tempered SMC samplers \citep{chopin2004, del2006sequential}. The consistency of $\hat{\rho}_t$ is an immediate application of Theorem~\ref{thm: posterior mean convergence} with $\varphi(\theta, x) = m_{\epsilon_t}(\sdp \mid x) / m_{\epsilon_{t-1}}(\sdp \mid x)$. The consistency of $\widehat{c_{\epsilon_T}(\sdp)}$ then follows by the continuous mapping theorem. Note that the integrability condition is mild. By Remark~\ref{remark: consistency assumption}, when $\Theta_t$ is compact and $w_t$ is bounded, it suffices that $m_{\epsilon_t} / m_{\epsilon_{t-1}}$ be bounded as a function of $x$. For the Laplace mechanism with $t>1$,
\begin{align*}
\frac{m_{\epsilon_t}(\sdp \mid x)}{m_{\epsilon_{t-1}}(\sdp \mid x)}
\le \frac{\epsilon_t}{\epsilon_{t-1}} \exp\left( \frac{\epsilon_{t-1}-\epsilon_t}{\Delta f} |\sdp - f(x)| \right) 
\le \frac{\epsilon_t}{\epsilon_{t-1}}
\end{align*}
is uniformly bounded in $x$ whenever $\epsilon_t \ge \epsilon_{t-1}$. Hence the condition is automatically satisfied under our DP-PF tempering schedule.

Algorithm~\ref{alg: DP-PF} already produces $\textup{\textbf{w}}^{(i,t)}$ (Line~12) and $x^{(i,t)}$ (Line~7) at each iteration, so the only extra computation is the log-density difference $\log m_{\epsilon_t}(\sdp \mid x^{(i,t)}) - \log m_{\epsilon_{t-1}}(\sdp \mid x^{(i,t)})$. Given $\widehat{c_{\epsilon_T}(\sdp)}$ for two competing models, the DP Bayes factor is
\begin{align*}
    \mathrm{BF}_{12}^{\mathrm{DP}} \;:=\; \frac{p_1(\sdp)}{p_2(\sdp)} \;\approx\; \frac{\widehat{c_{\epsilon_T}^{(1)}(\sdp)}}{\widehat{c_{\epsilon_T}^{(2)}(\sdp)}},
\end{align*}
which is well-defined provided both models use the same privatized statistic and privacy mechanism; the prior on $\theta$ and the data model $f$ may differ. Section~\ref{sec: simulation} demonstrates this on a Gaussian mixture model selection problem.

\section{Simulations} \label{sec: simulation}

In this section, we investigate applications of our approach proposed in Section~\ref{subsec: algorithm} to several common settings in DP Bayesian analysis, demonstrating the efficacy of the algorithm. 

\emph{General setup.} All resampling steps are conducted with multinomial resampling, and all perturbation kernels are Gaussian. Simulations were conducted on a computing cluster with multi-core compute nodes, 256 GB of memory, and a high-speed interconnect. We use 100 cores (1 node) per replicate per simulation with a maximum wall-time of 4 hours.

\subsection{Consistency and Posterior Inference} \label{subsec: simulation on consistency}

We demonstrate the consistency of Algorithm~\ref{alg: DP-PF} under a location-scale normal distribution, linear regression, and a Gaussian mixture model, using the Laplace mechanism and 1,000 replicates.

\subsubsection{Location-scale Normal} \label{sec: simulation of location-scale Gaussian}

Suppose the data model $f$ is $y_1,\ldots,y_n \ | \ \mu, \sigma^2 \overset{iid}{\sim} N \left( \mu, \sigma^2 \right)$ with prior distributions $\mu \sim N\left( m, \tau^2 \right)$ and $\sigma^2 \sim \text{Inv. Gamma} \left( \alpha, \beta \right)$, where $\alpha$ is the shape parameter and $\beta$ is the scale parameter. The private data $\tilde{y}_i$ is obtained by first clamping $y_i$ to the interval $[-5,5]$ and subsequently rescaling it to the interval $[-1,1]$. Let $L_1, L_2 \overset{iid}{\sim} \text{Lap} (\frac{\Delta}{\epsilon})$, where $\Delta = 3$. Then, the privatized statistics $\sdp$ are $(\sum \tilde{y}_i + L_1, \sum \tilde{y}_i^2 + L_2)$, which satisfies $\epsilon$-DP. 

\emph{Simulation setup.} True parameters are set as $\mu^* = 1, \sigma^*=1$, $\epsilon\in\{0.1,0.5,1,2\}$ and $n\in\{100,300,500,1000\}$. We generate a single value of $\sdp$ for each pair of $(\epsilon,n)$. The number of particles $N$ is set to achieve an average ESS of approximately 100 or higher. The hyperparameters $m = 0, \tau^2 = 4^2, \alpha = 1, \beta = 0.5$. We use an adaptive kernel with variance estimation and $T=2$ ($\epsilon_t = (0.5 \epsilon, \epsilon)$).

\begin{table}[t]
\centering
\resizebox{0.95\textwidth}{!}{%
\begin{tabular}{cc ccc ccc}
\toprule
&
&
\multicolumn{3}{c}{DP-PF} &
\multicolumn{3}{c}{DP-Reject-ABC} \\
\cmidrule(lr){3-5}
\cmidrule(lr){6-8}
$\epsilon$ & Sample Size
& $\ex(\mu \mid \sdp)$ & $\ex(\sigma^2 \mid \sdp)$ & Sec./ESS$^\dagger$
& $\ex(\mu \mid \sdp)$ & $\ex(\sigma^2 \mid \sdp)$ & Sec./ESS$^\dagger$ \\
\midrule
\multirow{4}{*}{0.1}
& 100  & 0.052 (0.0198) & 4.268 (0.3584) & 0.05 & 0.100 (0.0038) & 2.284 (0.2292) & 0.00 \\
& 300  & 1.014 (0.0146) & 3.902 (0.1316) & 0.05 & 1.722 (0.0020) & 2.215 (0.0112) & 0.00 \\
& 500  & 0.721 (0.0085) & 2.275 (0.1667) & 0.04 & 1.072 (0.0011) & 1.062 (0.0033) & 0.00 \\
& 1000 & 1.145 (0.0053) & 1.585 (0.0404) & 0.04 & 1.271 (0.0007) & 1.788 (0.0037) & 0.04 \\
\midrule
\multirow{4}{*}{0.5}
& 100  & 0.552 (0.0067) & 1.580 (0.0558) & 0.04 & 0.795 (0.0012) & 0.669 (0.0022) & 0.00 \\
& 300  & 1.123 (0.0036) & 0.988 (0.0105) & 0.03 & 1.171 (0.0005) & 1.233 (0.0022) & 0.03 \\
& 500  & 1.045 (0.0018) & 0.691 (0.0041) & 0.03 & 1.043 (0.0003) & 0.942 (0.0013) & 0.06 \\
& 1000 & 1.081 (0.0008) & 0.785 (0.0027) & 0.04 & 1.034 (0.0002) & 1.329 (0.0008) & 0.46 \\
\midrule
\multirow{4}{*}{1.0}
& 100  & 0.794 (0.0042) & 0.925 (0.0268) & 0.03 & 0.986 (0.0006) & 0.579 (0.0014) & 0.01 \\
& 300  & 1.126 (0.0015) & 0.729 (0.0038) & 0.03 & 1.096 (0.0003) & 1.130 (0.0013) & 0.09 \\
& 500  & 1.064 (0.0008) & 0.626 (0.0022) & 0.03 & 1.030 (0.0002) & 1.010 (0.0008) & 0.22 \\
& 1000 & 1.037 (0.0004) & 0.865 (0.0022) & 0.05 & 1.010 (0.0001) & 1.218 (0.0004) & 1.58 \\
\midrule
\multirow{4}{*}{2.0}
& 100  & 1.004 (0.0026) & 0.632 (0.0126) & 0.03 & 1.066 (0.0003) & 0.564 (0.0010) & 0.02 \\
& 300  & 1.099 (0.0006) & 0.660 (0.0022) & 0.03 & 1.062 (0.0002) & 1.069 (0.0007) & 0.29 \\
& 500  & 1.051 (0.0004) & 0.713 (0.0019) & 0.04 & 1.025 (0.0002) & 1.038 (0.0014) & 0.90 \\
& 1000 & 1.006 (0.0002) & 1.030 (0.0014) & 0.07 & 0.999 (0.0001) & 1.150 (0.0002) & 5.74 \\
\bottomrule
\end{tabular}%
}
\caption{Comparison of DP-PF and DP-Reject-ABC based on the averaged posterior means of $\mu$ and $\sigma^2$ over 1000 simulations for the location-scale normal model, and the seconds spent per ESS, with the values in parentheses representing MCSEs for each quantity. ($^\dagger$Parallelized over particles; 0.00 denotes runtimes $< 0.005$ seconds.)}
\label{table:type1}
\end{table}

Table~\ref{table:type1} compares the performance of DP-PF with DP-Reject-ABC, an exact sampling method. Across various privacy levels and sample sizes, DP-PF achieves comparable estimates for both $\mu$ and $\sigma$. The computation time per effective sample size is also similar. This demonstrates that DP-PF provides accuracy on par with the exact method while maintaining computational efficiency. While DP-Reject-ABC is effective in this setting, in the next section, we see that it does not scale well when applied to more complex settings.

\begin{figure}[t]
    \centering
    \includegraphics[scale=0.51]{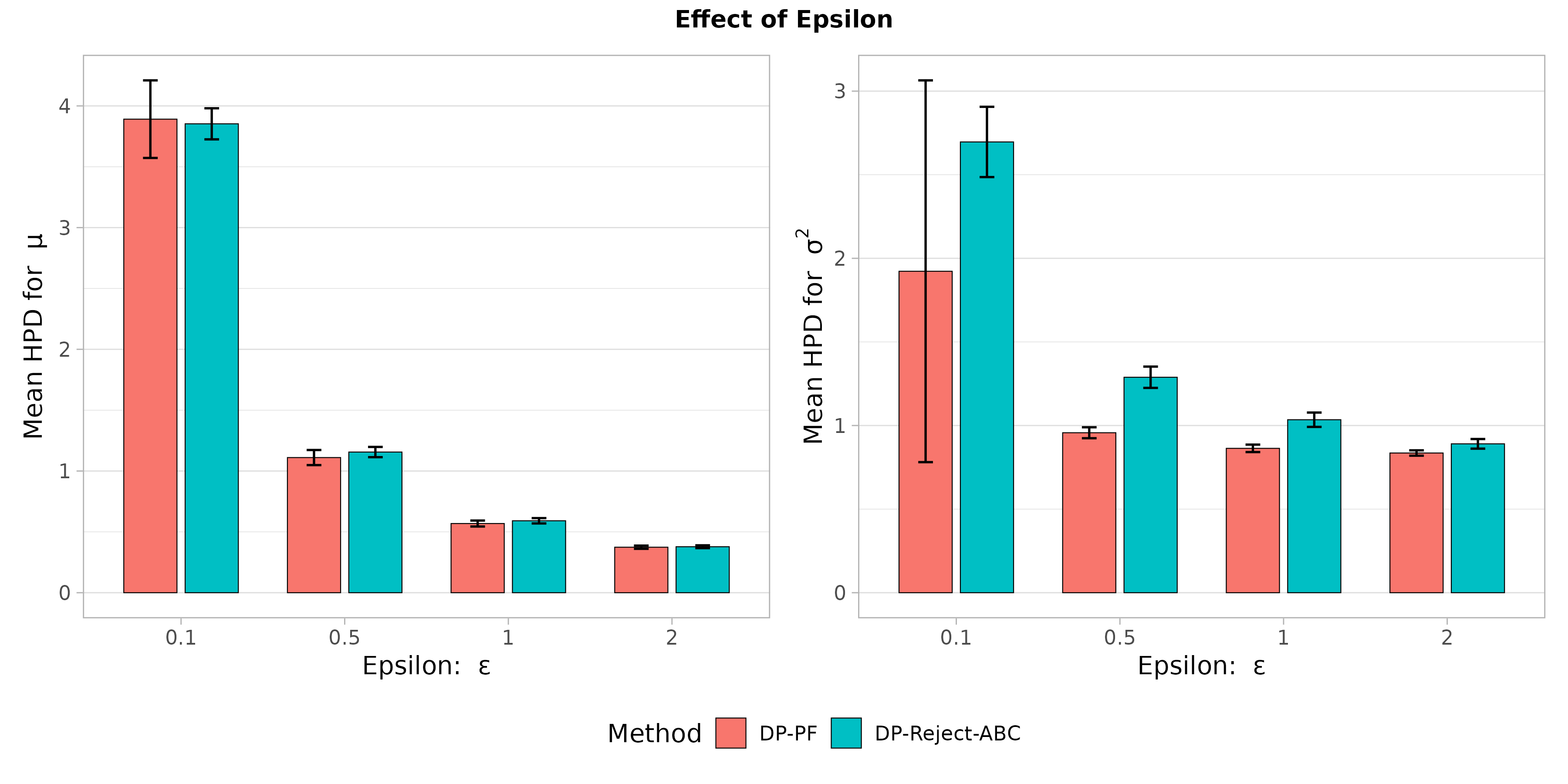}
    \caption{Comparison of the averaged lengths of the $90\%$ from DP-PF HPDs with the MCSEs of the posterior distributions of $\mu$ and $\sigma^2$ given $\sdp$, under different values of $\epsilon$.}
    \label{fig: HPD_epsilon}
\end{figure}

Figure~\ref{fig: HPD_epsilon} illustrates the $90\%$ highest posterior density (HPD) intervals for the posterior distributions of $\mu$ and $\sigma^2$. We again compare DP-PF with DP-Reject-ABC, which serves as an exact posterior sampler. The privacy budget $\epsilon$ ranges from $0.1$ to $2$, and DP-PF uses $1,000$ particles. Figure~\ref{fig: HPD_epsilon} shows that the HPD lengths produced by DP-PF exhibit greater variability when $\epsilon$ is small. As $\epsilon$ increases, the HPD lengths quickly become stable and closely match those obtained by DP-Reject-ABC, particularly for $\epsilon \geq 0.5$.

\subsubsection{Linear regression: Non-conjugate prior} \label{sec: lm:non-conjugate}

Suppose the data model $f$ is $y_i \mid x_i, \beta, \tau \sim N((1,x_i) \beta, \tau^{-1})$ and $x_i \mid \mu, \phi \sim N(\mu, \phi^{-1})$ with prior distributions $\beta \sim t_{2} \left( m, V, df \right)$, $\tau \sim \text{Weibull}\left( sc, sh \right)$, $\mu \sim t_{p}\left(  \theta=0_p, \Sigma=I_p, df=2 \right)$, and $\phi \sim \text{folded} \ t_{p}\left(  a=0_p, B=I_p, df=2 \right)$. The private data $\tilde{x}_i$ and $\tilde{y}_i$ are obtained by clamping $x_i$ and $y_i$ to $[-5,5]$ and subsequently rescaling them to $[-1,1]$. Let $L_1,..,L_5 \overset{iid}{\sim} \text{Lap} (\frac{\Delta}{\epsilon})$, where $\Delta = 8$ and $p=1$. Then, the privatized statistics $\sdp$ are $(\sum \tilde{y}_i + L_1, \sum \tilde{x}_i \tilde{y}_i + L_2, \sum \tilde{y}_i^2 + L_3, \sum \tilde{x}_i + L_4, \sum \tilde{x}_i^2 + L_5)$, which satisfies $\epsilon$-DP. In addition, the conjugate prior case is also implemented but deferred to \Cref{appendix: simulation}.
    
\emph{Simulation setup.} True values of the parameters are set as $\beta^* = (0, 2)^\top$, $\tau^* = 1$, $\mu^* = 1$, and $\phi^* = 1$. The number of particles $N=100,\,\epsilon\in\{0.5,1,2\}$ and $n\in\{300,500\}$. The hyperparameters are set as $m=0_{2}, V=I_{2}, df=2, sc=1.25, sh=2, \theta=0, \Sigma=1, a=0, B=1$. We use a kernel with a pre-specified scale and $T=10$ ($\epsilon_t = (0.1 \epsilon, 0.2 \epsilon,\ldots, \epsilon)$); results with different schedules are included in Section~\ref{appendix: simulation}.

\begin{table}[t]
\centering
\resizebox{0.9\textwidth}{!}{%
\begin{tabular}{cc cc cc cc}
\toprule
&
&
\multicolumn{2}{c}{DP-PF} &
\multicolumn{2}{c}{DP-DA-MCMC} &
\multicolumn{2}{c}{DP-Reject-ABC} \\
\cmidrule(lr){3-4}
\cmidrule(lr){5-6}
\cmidrule(lr){7-8}
$\epsilon$ & Sample Size
& $\ex(\beta_1 \mid \sdp)$ & Sec./ESS$^\dagger$
& $\ex(\beta_1 \mid \sdp)$ & Sec./ESS
& $\ex(\beta_1 \mid \sdp)$ & Sec./ESS$^\dagger$ \\
\midrule
\multirow{2}{*}{0.5}
& 300 & 1.973 (0.0128)    & 0.293 & 1.871 (0.0324)    & 629.809  & 2.083 (0.0031)    & 0.257 \\
& 500 & 1.772 (0.0079)    & 0.638 & 1.818 (0.0187)    & 1173.600 & 1.849 (0.0023)    & 7.079 \\
\midrule
\multirow{2}{*}{1.0}
& 300 & 1.983 (0.0085)    & 0.250 & 1.957 (0.0167)    & 370.535  & 2.116 (0.0021)    & 2.053 \\
& 500 & 1.889 (0.0070)    & 1.696 & 1.933 (0.0107)    & 857.429  & 1.934 (0.0016)    & 70.577 \\
\midrule
\multirow{2}{*}{2.0}
& 300 & 2.048 (0.0076)    & 0.634 & 2.075 (0.0113)    & 255.562  & 2.142 (0.0015)    & 23.235 \\
& 500 & 1.980 (0.0061)    & 7.481 & 1.995 (0.0071)    & 605.773  & $-$               & $-$ \\
\bottomrule
\end{tabular}
}
\caption{Comparison of DP-PF, DP-DA-MCMC, and DP-Reject-ABC based on the averaged posterior mean of $\beta_1$ and seconds per effective sample size (Sec./ESS) over 1000 simulations for linear regression under a non-conjugate prior, with values in parentheses representing MCSEs for the posterior mean estimates. ($^\dagger$ Parallelized over particles.)}
\label{table:type2_nonconjugate}
\end{table}

Table~\ref{table:type2_nonconjugate} demonstrates that DP-PF maintains accurate posterior mean estimation across a range of sample sizes and privacy budgets, even under a challenging heavy-tailed $t$ prior with two degrees of freedom, whose variance is undefined. This prior is especially difficult for Monte Carlo-based DP methods because clamping may provide insufficient information to correct proposals far from the observed data. Since $\sdp$ is fixed for each $(n,\epsilon)$ combination, posterior means may vary across settings when a ``bad seed'' generates an unrepresentative $\sdp$. Nevertheless, DP-PF remains computationally efficient. In contrast, DP-DA-MCMC has approximately twice the MCSEs and requires roughly 100 times more computation due to the additional Metropolis--Hastings step and its inability to parallelize. Although DP-Reject-ABC is theoretically exact, it failed to complete within 4 hours for $n=500$ and $\epsilon=2$ and is therefore omitted. \Cref{appendix: simulation} reports analogous results under conjugate priors, where DP-PF still outperforms DP-DA-MCMC under parallelization, although by a smaller margin. We also report the performance and computation time of DP-PF under an exponential $\epsilon_t$ schedule.

Table~\ref{table:type2_nonconjugate_hpd} shows that DP-PF produces $90\%$ HPD widths that are consistently closer to those of DP-Reject-ABC than those produced by DP-DA-MCMC. This indicates that DP-PF more accurately captures posterior distributions in addition to providing accurate point estimates. The advantage is particularly pronounced when the privacy budget is moderate or large, where DP-PF closely tracks the benchmark while retaining its computational efficiency.

\begin{table}[t]
\centering
\resizebox{0.65\textwidth}{!}{%
\begin{tabular}{cc ccc}
\toprule
&
&
DP-PF & DP-DA-MCMC & DP-Reject-ABC \\
\cmidrule(lr){3-3}
\cmidrule(lr){4-4}
\cmidrule(lr){5-5}
$\epsilon$ & Sample Size
& 90\% HPD & 90\% HPD & 90\% HPD \\
\midrule
\multirow{2}{*}{0.5}
& 300 & 2.499 (0.0139) & 4.485 (0.0625)  & 2.740 (0.0101) \\
& 500 & 2.097 (0.0120) & 3.009 (0.0274)  & 2.186 (0.0078) \\
\midrule
\multirow{2}{*}{1.0}
& 300 & 1.935 (0.0093) & 2.932 (0.0318)  & 1.864 (0.0063) \\
& 500 & 1.435 (0.0072) & 1.918 (0.0176)  & 1.362 (0.0042) \\
\midrule
\multirow{2}{*}{2.0}
& 300 & 1.388 (0.0059) & 1.899 (0.0189)  & 1.244 (0.0038) \\
& 500 & 0.947 (0.0045) & 1.164 (0.0084)  & $-$ \\
\bottomrule
\end{tabular}%
}
\caption{Comparison of DP-PF, DP-DA-MCMC, and DP-Reject-ABC based on the averaged 90\% HPD widths of $\beta_1$ over 1,000 simulations for linear regression under a non-conjugate prior, with values in parentheses representing MCSEs.}
\label{table:type2_nonconjugate_hpd}
\end{table}

\subsubsection{Gaussian mixture model} \label{sec: Gaussian mixture}

Suppose the data model $f$ is $x_1, \ldots, x_n \mid \mu_1, \mu_2, \sigma_1^2, \sigma_2^2, p \overset{iid}{\sim} p N(\mu_1, \sigma_1^2) + (1-p) N(\mu_2, \sigma_2^2)$ with prior distributions $\mu_1, \mu_2 \overset{iid}{\sim} N(m, \tau^2)$ with the identifiability constraint $\mu_1 < \mu_2$ enforced by taking the minimum and maximum of two independent draws, $\sigma_k^2 \sim \text{Inv. Gamma}(\alpha, \beta)$ for $k = 1, 2$, and $p \sim \text{Beta}(a, b)$. Let $L_1, \ldots, L_K \overset{iid}{\sim} \text{Lap}\left(\frac{M}{n\epsilon}\right)$, where $K = \sum_{j=1}^{M} 2^j = 2^{M+1} - 2$ is the total number of histogram bins across all $M$ layers. The privatized statistic $\sdp$ is the $M$-layer hierarchical histogram, where layer $j$ partitions the clamping region $[L, U]$ into $2^j$ equal bins and each layer receives an equal privacy budget of $\epsilon/M$, giving a Laplace noise scale of $M/(n\epsilon)$ per layer. This satisfies $\epsilon$-DP by the composition property.

\emph{Simulation setup.} True parameters are set as $\mu_1^* = 1$, $\mu_2^* = 3$, $\sigma_1^{*2} = \sigma_2^{*2} = 0.25$, $p^* = 0.4$, with clamping bounds $L = 0$ and $U = 5$. We use $M \in \{2, 3\}, \epsilon \in \{0.5, 1, 2\}, n \in \{100, 200, 500, 1000\}$, and $N = 100$ particles. The hyperparameters are $m = 2.5$, $\tau^2 = 1.5^2$, $\alpha = 3$, $\beta = 0.5$, $a = b = 1$. We use an adaptive scaling of the perturbation kernel based on the covariance structure of particles in each iteration and $T=10$ ($\epsilon_t = (0.1 \epsilon, 0.2 \epsilon,\ldots, \epsilon)$).

\begin{table}[t]
\centering
\resizebox{0.95\textwidth}{!}{
\begin{tabular}{ccc ccccc}
\toprule
&
&
&
\multicolumn{5}{c}{DP-PF} \\
\cmidrule(lr){4-8}
$\epsilon$ & Layers & Sample Size
& $\ex(\mu_1 \mid \sdp)$ & $\ex(\mu_2 \mid \sdp)$ & $\ex(\sigma_1^2 \mid \sdp)$ & $\ex(\sigma_2^2 \mid \sdp)$ & $\ex(p \mid \sdp)$ \\
\midrule
\multirow{8}{*}{0.50} & \multirow{4}{*}{2}
& 100  & 1.218 (0.114) & 3.044 (0.089) & 0.812 (0.102) & 0.748 (0.103) & 0.475 (0.044) \\
& & 200  & 0.961 (0.068) & 2.927 (0.039) & 0.510 (0.075) & 0.349 (0.030) & 0.419 (0.022) \\
& & 500  & 0.921 (0.027) & 2.962 (0.015) & 0.347 (0.021) & 0.259 (0.011) & 0.373 (0.009) \\
& & 1000 & 0.913 (0.016) & 2.997 (0.012) & 0.294 (0.012) & 0.268 (0.007) & 0.399 (0.006) \\
\cmidrule(lr){2-8}
& \multirow{4}{*}{3}
& 100  & 1.264 (0.121) & 2.967 (0.092) & 0.835 (0.109) & 0.724 (0.105) & 0.482 (0.045) \\
& & 200  & 1.105 (0.049) & 2.966 (0.035) & 0.464 (0.067) & 0.307 (0.037) & 0.457 (0.022) \\
& & 500  & 0.944 (0.022) & 2.995 (0.012) & 0.419 (0.030) & 0.272 (0.012) & 0.389 (0.009) \\
& & 1000 & 0.894 (0.011) & 2.959 (0.008) & 0.287 (0.012) & 0.277 (0.007) & 0.386 (0.005) \\
\bottomrule
\end{tabular}
}
\caption{Average posterior estimates from DP-PF for the Gaussian Mixture Model parameters $\mu_1, \mu_2, \sigma_1^2, \sigma_2^2$, and $p$ across 1,000 simulations, with the values in parentheses representing the MCSEs for each quantity.}
\end{table}

\begin{table}[t]
\centering
\resizebox{0.95\textwidth}{!}{%
\begin{tabular}{ccc ccccc}
\toprule
&
&
&
\multicolumn{5}{c}{DP-DA-MCMC} \\
\cmidrule(lr){4-8}
$\epsilon$ & Layers & Sample Size
& $\ex(\mu_1 \mid \sdp)$ & $\ex(\mu_2 \mid \sdp)$ & $\ex(\sigma_1^2 \mid \sdp)$ & $\ex(\sigma_2^2 \mid \sdp)$ & $\ex(p \mid \sdp)$ \\
\midrule
\multirow{8}{*}{0.5} & \multirow{4}{*}{2}
& 100  & 0.999 (0.083) & 3.019 (0.036) & 1.196 (0.155) & 0.348 (0.040) & 0.577 (0.015) \\
& & 200  & 0.869 (0.097) & 2.829 (0.023) & 0.817 (0.100) & 0.158 (0.005) & 0.506 (0.011) \\
& & 500  & 0.607 (0.172) & 2.915 (0.030) & 1.053 (0.289) & 0.156 (0.009) & 0.485 (0.022) \\
& & 1000 & 0.491 (0.277) & 2.954 (0.062) & 1.138 (0.424) & 0.183 (0.014) & 0.532 (0.039) \\
\cmidrule(lr){2-8}
& \multirow{4}{*}{3}
& 100  & 1.241 (0.026) & 2.961 (0.011) & 0.591 (0.053) & 0.231 (0.014) & 0.536 (0.010) \\
& & 200  & 1.256 (0.008) & 2.951 (0.005) & 0.431 (0.014) & 0.147 (0.003) & 0.507 (0.004) \\
& & 500  & 0.856 (0.039) & 3.026 (0.006) & 1.116 (0.105) & 0.202 (0.005) & 0.485 (0.010) \\
& & 1000 & 0.891 (0.002) & 2.955 (0.002) & 0.340 (0.012) & 0.255 (0.002) & 0.395 (0.002) \\
\bottomrule
\end{tabular}%
}
\caption{Average posterior estimates from DP-DA-MCMC for the Gaussian Mixture Model parameters $\mu_1, \mu_2, \sigma_1^2, \sigma_2^2$, and $p$ across 1,000 simulations, with the values in parentheses representing the MCSEs for each quantity.}
\end{table}

For the case of $M=3$, both DP-PF and DP-DA-MCMC perform well, as the degrees of freedom of $\sdp$ exceed the number of parameters. For the case of $M=2$, this is not the case, and because of this, we expect the posterior distribution to be multi-modal. In this case, DP-PF still recovers sufficient information from $\sdp$ and yields accurate posterior estimates for all parameters, whereas DP-DA-MCMC recovers only one location parameter well.

\subsection{Asymptotic Confidence Interval Coverage} \label{subsec: simulation on asymptotic CI}

In this section, we demonstrate using the linear regression model under both conjugate and non-conjugate priors, introduced in Section~\ref{subsec: simulation on consistency} and Appendix~\ref{subsec: linear regression with conjugate priors}, that Algorithm~\ref{alg: DP-PF} achieves the nominal $95\%$ coverage of its asymptotic confidence interval. The values of $\sdp$ are carefully chosen so that the private posterior mean of $\beta_1$ is known to be $0$, which we use as the ground truth for assessing coverage.

\emph{Simulation setup.} We set $L=-5$, $U=5$ and the privatized statistics $\sdp = (0,0,50,0,50)$, in the same order as in Section~\ref{sec: lm:non-conjugate}. $n=500$, $\epsilon\in\{0.1,0.5,1,2\}$ and $N\in\{100,200,300,500,\\1000,1200,1400,1600\}$. We use 100 independent simulations to assess whether the empirical coverage is consistent with the nominal $95\%$ level.

Note that there are no true parameters in this setting, as the data are not generated from any model. Instead, we construct $\sdp$ to have a convenient property: the posterior mean of $\beta_1$ is known to be $0$ under the symmetric and well-behaved prior. This construction of $\sdp$ allows us to evaluate coverage of the posterior mean, which would be difficult for a generic private summary statistic, as the posterior mean usually does not have a closed-form expression.

\begin{figure}[t]
    \centering
    \includegraphics[scale=0.51]{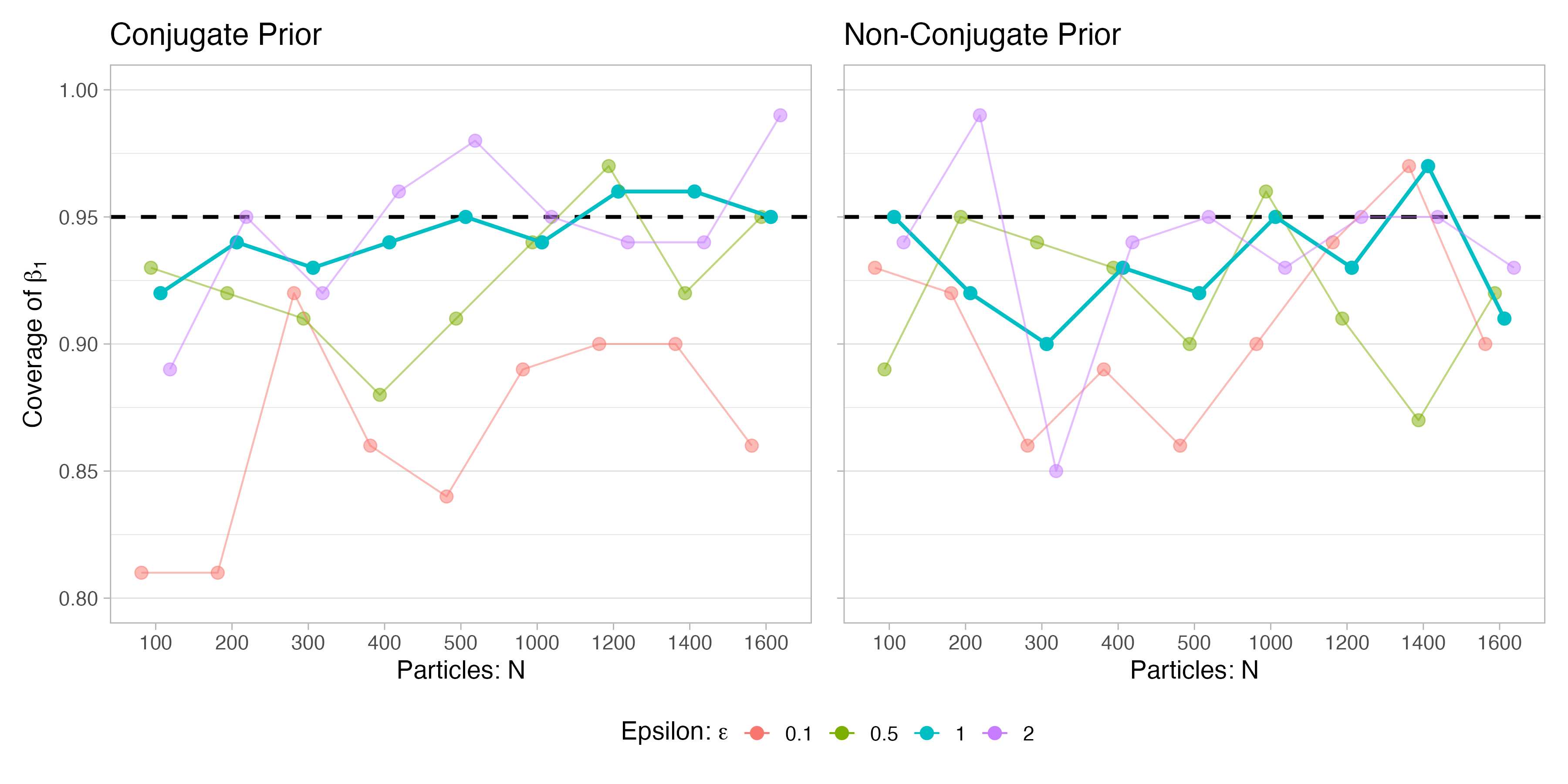}
    \caption{Comparison of conjugate and non-conjugate priors on the $\beta_1$ coverage in $N$ for linear regression via DP-PF. The estimates are based on Theorem~\ref{thm: CI}, which establishes $95\%$ asymptotic confidence intervals for the posterior mean of $\beta_1$.}
    \label{fig: beta1 coverage in N}
\end{figure}

Figure~\ref{fig: beta1 coverage in N} illustrates that the coverage of the $95\%$ confidence interval (CI) for the posterior mean of $\beta_1$ generally approaches the nominal level as the number of particles $N$ increases. For $\epsilon \geq 0.5$, the coverage stabilizes around 0.95 when $N > 1000$, indicating that DP-PF achieves reliable coverage with a sufficiently large number of particles. Moreover, as $N$ increases, the coverage improves while the interval width decreases. Notably, smaller values of $\epsilon$ require a larger $N$ to achieve the same coverage and interval width as larger $\epsilon$. This observation aligns with the findings for DP-DA-MCMC, which noted worse mixing of their chain for small $\epsilon$. This, however, is not a major concern for our method, as smaller $\epsilon$ values allow for faster computations, making it feasible to use a larger $N$ without significant computational overhead.

\begin{figure}[t]
    \centering
    \includegraphics[scale=0.51]{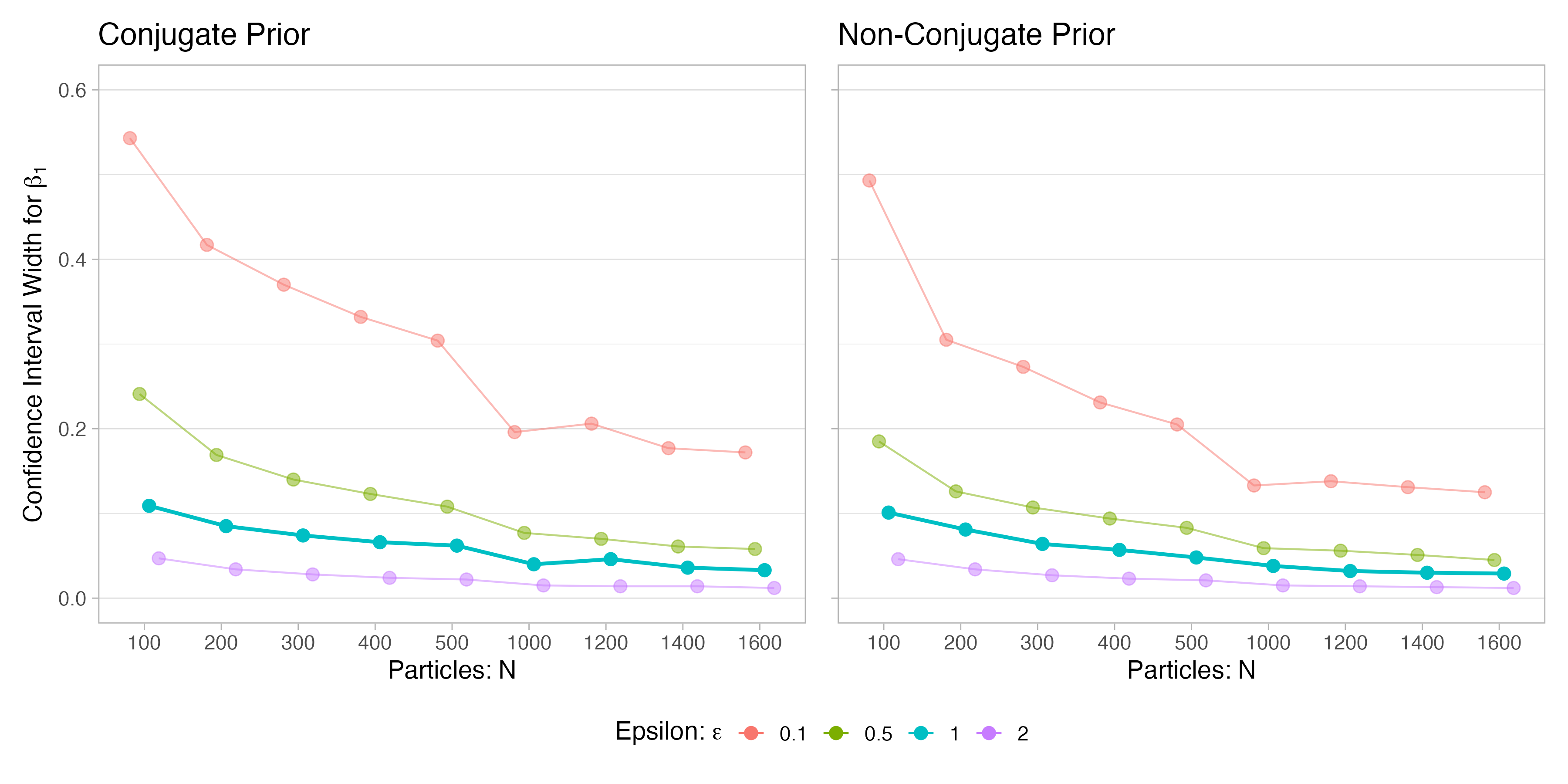}
    \caption{Comparison of conjugate and non-conjugate priors on the confidence interval width of $\beta_1$ in $N$ for linear regression via DP-PF. The estimates are based on Theorem~\ref{thm: CI}, which provides $95\%$ asymptotic confidence intervals for the posterior mean of $\beta_1$.}
    \label{fig: CI width of beta1 in N}
\end{figure}

Figure~\ref{fig: CI width of beta1 in N} shows that the average CI width decreases as $N$ increases. For $\epsilon \geq 0.5$, the CI width stabilizes below 0.05 when $N > 1000$, indicating that DP-PF achieves both precise and stable inference with a sufficient number of particles. This suggests that increasing the particle size beyond this point provides limited additional improvement in the precision of the resulting inference in this case.

\subsection{Model Evidence}

We consider a model comparison problem where three candidate Bayesian models are fit to the same privatized summary statistic $\sdp$, using 1,000 replicates. The goal is to assess which model best explains $\sdp$ under a fixed privacy mechanism.

The three competing models are: (M1) a single Gaussian model (Section~\ref{sec: Gaussian mixture}); (M2) a two-component Gaussian mixture model (Section~\ref{sec: lm:non-conjugate}); and (M3) a three-component Gaussian mixture model:
\begin{align*}
x_i \mid \underline{\mu}, \underline{\sigma}^2, \underline{p} \overset{iid}{\sim} \sum_{k=1}^3 p_k  N(\mu_k, \sigma_k^2),
\end{align*}
with ordered means $\mu_1 < \mu_2 < \mu_3$, priors $\mu_k \sim N(m,\tau^2)$, $\sigma_k^2 \sim \mathrm{Inv\text{-}Gamma}(\alpha,\beta)$, and $(p_1,p_2,p_3)\sim \mathrm{Dir}(1,1,1)$. The true data-generating model is a 2-component Gaussian mixture (M2), and the privacy mechanism is the hierarchical histogram queries used in Section~\ref{sec: Gaussian mixture} with $M=2$.

\emph{Simulation setup.} All models are fit using the DP-PF algorithm with identical mechanisms, privacy budget, and hyperparameters as described above, with $n=100$, $\epsilon=0.5$, $N=100$ particles, and $T=10$ iterations (linear $\epsilon_t$ schedule). All three models use the same hyperparameters $m=2.5$, $\tau^2=1.5^2$, $\alpha=3$, $\beta=0.5$ for mean and variance parameters; M2 additionally sets $a=b=1$, while M3 uses a symmetric Dirichlet prior. Data are generated from a Gaussian mixture model with two components; that is, M2 is the true model. Results are summarized over 1,000 independent replications in Table~\ref{tab:model_comparison}, where $\hat{P}(\text{best model})$ denotes the proportion of replications in which each model achieves the highest log evidence.

\begin{table}[t]
\centering
\resizebox{0.8\textwidth}{!}{
\begin{tabular}{lccc}
\toprule
& M1 (Single Gaussian) & M2 (2-component) & M3 (3-component) \\
\midrule
Mean log evidence & 3.886 (0.0041) & 4.820 (0.0048) & 4.879 (0.0063) \\
$\hat{P}(\text{best model})$ & 0.000 (0.0000) & 0.397 (0.0154) & 0.603 (0.0154) \\
\bottomrule
\end{tabular}}
\caption{Model comparison via DP-PF marginal likelihood estimates over 1{,}000 replicates.}
\label{tab:model_comparison}
\end{table}

Table~\ref{tab:model_comparison} shows that M2 is selected as the best model in approximately $40\%$ of the replications, while M3 is selected in about $60\%$ of the time; M1 is never selected as the best model across all replications, hence its zero MCSE. Moreover, the log evidence of M2 and M3 is close relative to their MCSEs, indicating that neither model is strongly favored over the other.

\section{Data Analysis} \label{sec: data analysis}

In this section, we demonstrate that our particle filtering algorithm can be applied to a non-additive privacy mechanism as well. In particular, we derive the acceptance probability for a logistic regression, using the objective perturbation mechanism \citep{chaudhuri2011differentially} with the formulation given in \citet{awan2021structure}, which adds noise to the gradient of the log-likelihood before minimization, thereby creating non-additive noise. Moreover, since we assume that the independent variable is protected, we privatize its first two moments via the K-norm mechanism \citep{hardt2010geometry, awan2021structure}, a multivariate additive-noise mechanism, which is similar to the setup used in \citet{awan2025simulation}.

We conduct an experiment with the 2021 Census Public Use Microdata Files (PUMF), which provide data on the characteristics of the Canadian population \citep{StatisticsCanada2022}. We analyze the dependence between age and retirement status by inferring logistic regression under DP guarantees.

\subsection{Experiment Setup}

The PUMF data set contains 980,868 records of individuals, representing $2.5\%$ of the Canadian population. Among the 144 variables in the data set, we chose three variables: the census metropolitan area or census agglomeration of current residence (named CMA), the age (named AGEGRP), and the private retirement income (named Retir). We extracted the records of AGEGRP and Retir belonging to the people in Victoria (according to the values in CMA). After pre-processing the records with unavailable or not applicable values, the sample size is 10,473; Retir is transformed into a binary variable with Retir = 1 for retirement and Retir = 0 otherwise, and AGEGRP is scaled into the range of $[0,1]$. For the preliminary data analysis, we find the non-private regression coefficient between Retir and AGEGRP is $\hat{\beta} = (-3.825, 6.822)$, and the distribution of AGEGRP is near-uniform with softened edges.

The DP-PF algorithm is conducted with multinomial resampling and Gaussian perturbation kernels. We use a kernel with pre-specified scale and $T=10$ ($\epsilon_t = (0.1 \epsilon, 0.2 \epsilon,\ldots, \epsilon)$ and $q_t = (0.02 q, 0.02 q,\ldots,0.02q, q)$). $N=200$. This additional $q_t$ schedule creates a similar effect to the $\epsilon_t$ schedule, contributing to flatter intermediate target distributions and thereby increasing the acceptance probability $r_t$. The same computing environment is used as in Section~\ref{sec: simulation}.

\subsection{Logistic Regression} \label{subsec: logistic regression}

The logistic regression model is assumed to be $P(Y=1\mid X) = \frac{1}{1+\exp\{-(\beta_0+\beta_1X)\}}$ and $Z \overset{d}{=} \frac{X+1}{2} \sim \text{Beta}(a,b)$, with prior distributions $\beta_0,\beta_1 \overset{\mathrm{iid}}{\sim} N(0,16)$ and $a,b \overset{\mathrm{iid}}{\sim} \text{Gamma}(6,4)$. We set $\epsilon=0.5$ and $q=0.5$, and randomly draw $n=1000$ samples from the Victoria data set. The $\sdp=(\beta_{0_{dp}}, \beta_{1_{dp}}, \sum z_i + K_1, \sum z_i^2 + K_2)$ we get from objective perturbation and K-norm mechanisms (Algorithms~\ref{alg: sampling from ell_infty} and \ref{alg: ell_infty objective perturbation}) is (-2.982634, 5.119215, 479.573100, 302.947785). We derive the acceptance probability of our algorithm on the objective perturbation for the regression coefficient and introduce the K-norm mechanism for the first two moments' sum of AGEGRP in \Cref{appendix: real data analysis}. We use the $\ell_\infty$-norm with $\Delta_\infty = 2$ for both mechanisms, allocating $90\%$ of the privacy budget to the regression coefficients and $10\%$ to the two moments. Our method gives the private posterior mean of $\beta_0 = -2.357$ and $\beta_1 = 3.816$ with an effective sample size of 150.787 and runtime of 3103.811 seconds. 

\begin{figure}[htbp]
    \centering
    \includegraphics[scale=0.15]{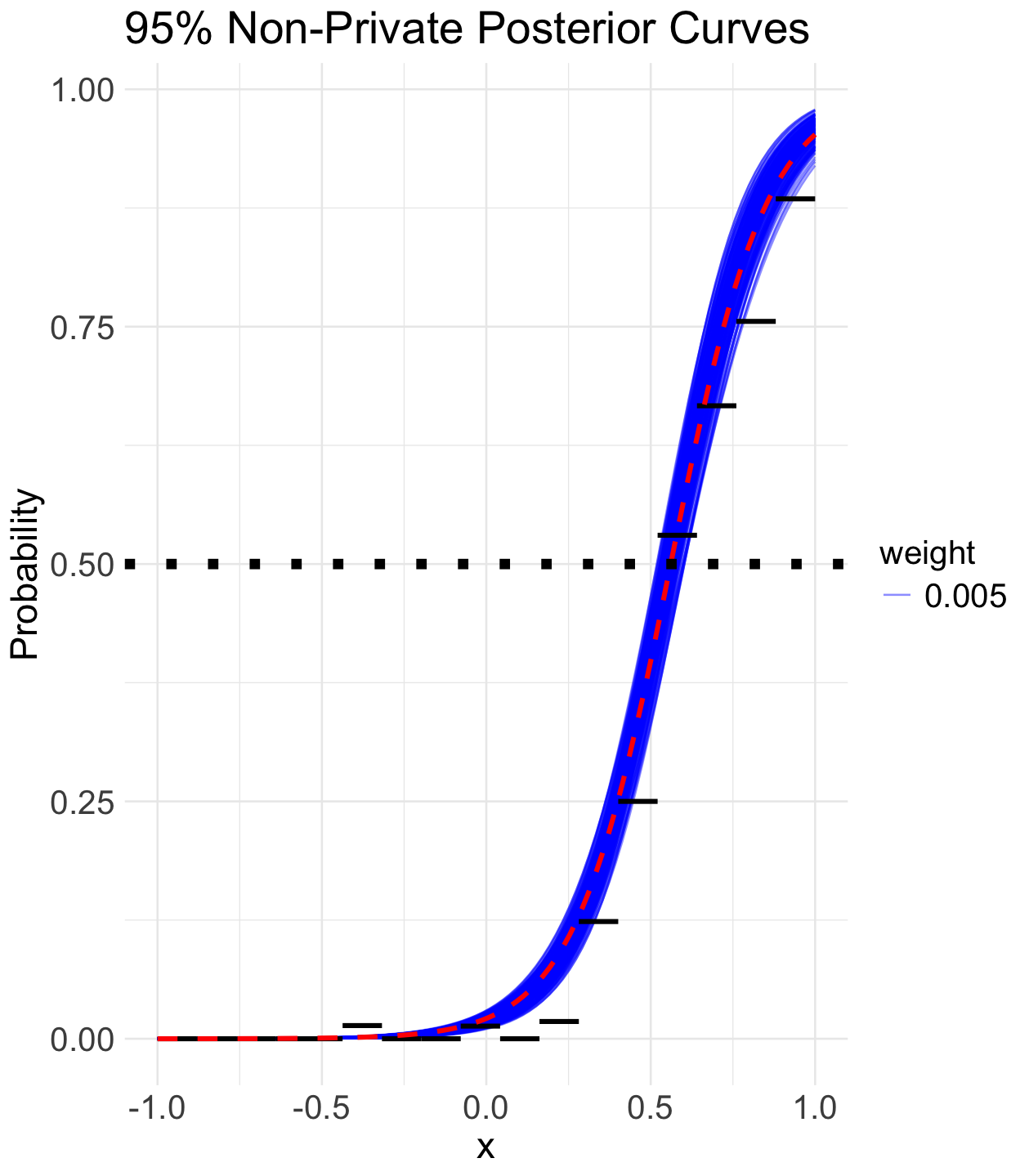}
    \includegraphics[scale=0.15]{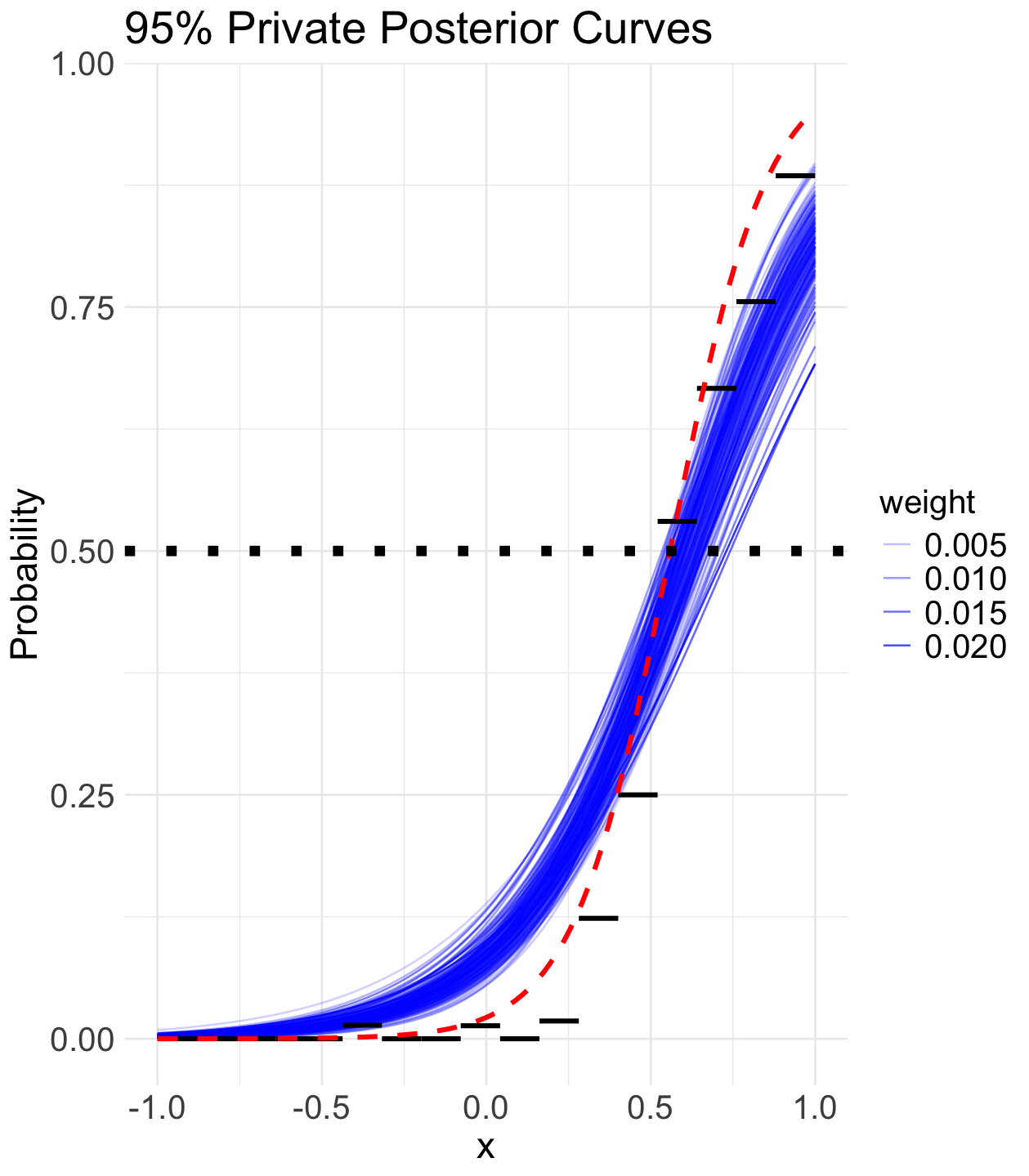}
    \caption{Comparison of $95\%$ non-private and private posterior curves for logistic regression via Algorithms~\ref{alg: sampling from ell_infty} and \ref{alg: ell_infty objective perturbation}, which illustrates the result of Section~\ref{subsec: logistic regression}. Black step lines: proportions of retirement from AGEGRP. Red dashed curve: fitted logistic curve. Horizontal dotted line: model given by the prior means.}
    \label{fig: 0.95 private curves}
\end{figure}

Figure~\ref{fig: 0.95 private curves} compares the $95\%$ most representative private curves with the $95\%$ non-private posterior curves. The selected private curves correspond to the sampled $\beta$ whose values are closest to the median, as measured by the Mahalanobis distance using the covariance matrix of $\beta$. Similarly, the selected non-private curves, obtained via the Bayesian Metropolis-Hastings algorithm, correspond to the $95\%$ closest posterior samples to its median among the 200 thinned $\beta$ samples, drawn from a total of 5,000 iterations after a 5,000-sample burn-in, with an acceptance rate of $20.72\%$. The blue curves represent the estimated logistic functions, while the black step lines indicate the observed proportions of retirement across age groups. The red dashed curve shows the logistic fit based on the original (non-private) data, and the horizontal dotted line represents the model determined by the prior means.

While the non-private posterior curves closely align with the red fitted curve, it is critical to note that the private estimation is based solely on differentially private (DP) summary statistics, without access to the original data. Despite some variability among the private curves, their alignment with the black step lines suggests that the privatized estimates capture the underlying data structure reasonably well. Although DP introduces randomness at the cost of information loss---causing our method to be more susceptible to regularization effects from the prior---the private logistic curves approximate the trend and shape of the non-private fit, maintaining a reasonable correspondence with the observed proportions and preserving key patterns in the data while ensuring privacy.

\section{Discussion and Directions for Future Work} \label{sec: conclusions}

In this paper, we propose a novel differentially private particle filtering algorithm to generate approximate posterior samples based on confidential queries. Our sampler offers the significant advantage of producing consistent estimates for the parameter $\theta$ given $\sdp$. The tempering $\epsilon_t$ schedule enhances the efficiency of the sampler. Compared to MCMC methods, our sampler does not require strong assumptions for convergence or for (CLT) to hold. 

Our method is conceptually related to existing approaches that treat the differentially private posterior as doubly intractable, but differs in assumptions, computation, and practical behavior. Compared with DP-DA-MCMC of \citet{ju2022data}, which augments the confidential database and relies on a privacy-aware Metropolis-within-Gibbs sampler, DP-PF replaces Markov chain dynamics with sequential Monte Carlo, avoiding mixing issues, multimodality concerns, and the geometric ergodicity assumptions required for theoretical guarantees in DP-DA-MCMC, while accommodating non-conjugate priors more naturally. At the same time, because DP-DA-MCMC updates latent data sequentially, it may scale more favorably with sample size $n$ in some settings. Compared with DP-Reject-ABC of \citet{gong2022exact}, which exploits additive-noise structure through rejection sampling and may suffer from low acceptance rates or limited applicability to other perturbation mechanisms, DP-PF combines a privacy-tempering schedule with sequential propagation and rejection correction to improve sampling efficiency while preserving the exactness of the target distribution. Empirically, our simulations demonstrate competitive or improved computational performance and accurate posterior inference, while the SMC framework additionally provides Monte Carlo error quantification and marginal likelihood estimation.

While our analysis in Section~\ref{sec: particle filters for DP} focuses on multinomial resampling, the validity of our Monte Carlo estimates and confidence intervals extends to other resampling schemes as well. Extending the proof to stratified or systematic resampling is not straightforward, as their next-iteration particle systems are determined by a single uniform random draw, unlike multinomial and residual resampling, which generate independent particles. However, all these resampling schemes satisfy a central limit theorem, ensuring that the estimates and intervals remain valid under the appropriate asymptotic variances. Although different schemes may alter the form of the asymptotic variance, our Taylor expansion-based analysis still applies. Moreover, some schemes may achieve smaller asymptotic variances compared to multinomial resampling, depending on how they influence the resampled particle system. The detailed analysis is left for future work.

Our algorithm inevitably becomes more computationally demanding as the sample size and privacy budget $\epsilon$ increase, and its performance depends critically on the quality of the importance weights. In addition, the dimensions of both the parameter space $\theta$ and the privatized statistic $\sdp$ introduce the curse of dimensionality common to importance sampling methods \citep{Chopin_Papaspiliopoulos_2020}. Because the importance weights must be analytically tractable, some potentially efficient proposal distributions may be excluded. Another limitation is that each iteration generates synthetic data solely from $\theta$, without leveraging information from previously generated data. Our empirical results also suggest that stable and reliable inference typically requires at least 50 observations per parameter, indicating that high-dimensional inference with our DP-PF remains fundamentally difficult. Future work may explore dimension-reduction techniques \citep{talwar2015lasso} or efficient synthetic-data update approaches such as DP-DA-MCMC \citep{ju2022data} to improve scalability in higher-dimensional settings.

An additional direction for future work is to improve the scalability and adaptivity of the DP-PF framework. A promising extension is adaptive tempering, where the increment of $\epsilon_t$ is determined dynamically based on particle degeneracy measures such as the effective sample size (ESS), following adaptive SMC ideas in \citet{agapiou2017importance, Chopin_Papaspiliopoulos_2020}. Such strategies may improve computational stability and efficiency by deciding the number of iterations of DP-PF and allocating its tempering scheme based on the evolving particle system, but they also introduce theoretical challenges because tempering schedules are data-dependent, complicating privacy accounting and the statistical analysis. As for high-dimensional problems, the synthetic data updates can further bottleneck parameter updates. While projecting updates onto a lower-dimensional feature space—by leveraging sufficient statistics or latent factors—offers a scalable and efficient solution, it is highly problem-specific and requires additional assumptions.

\acks{}
This work was supported in part by NSF grants SES-2150615 and SES-2610910. Part of this project was completed while Jordan Awan was a member of the Department of Statistics at Purdue University. The idea to manipulate the privacy budget in our sampler was inspired by preliminary work completed by a team consisting of Ruobin Gong, Nianqiao Ju, Vinayak Rao, and Jordan Awan, where a similar $\epsilon$-schedule was used in an MCMC sampler; we are grateful for these prior discussions that led to this work.

\newpage

\appendix

\section{More Results of Simulation} \label{appendix: simulation}

In this section, we present additional simulation results that could not be included in the main text but provide further support for our findings.

\subsection{Location-Scale Normal: HPD in sample sizes}

To examine the effect of sample size on the HPD intervals, we run DP-PF across a range of sample sizes. All other settings are identical to those used in Figure~\ref{fig: HPD_epsilon}.

\begin{figure}[htbp]
    \centering
    \includegraphics[scale=0.47]{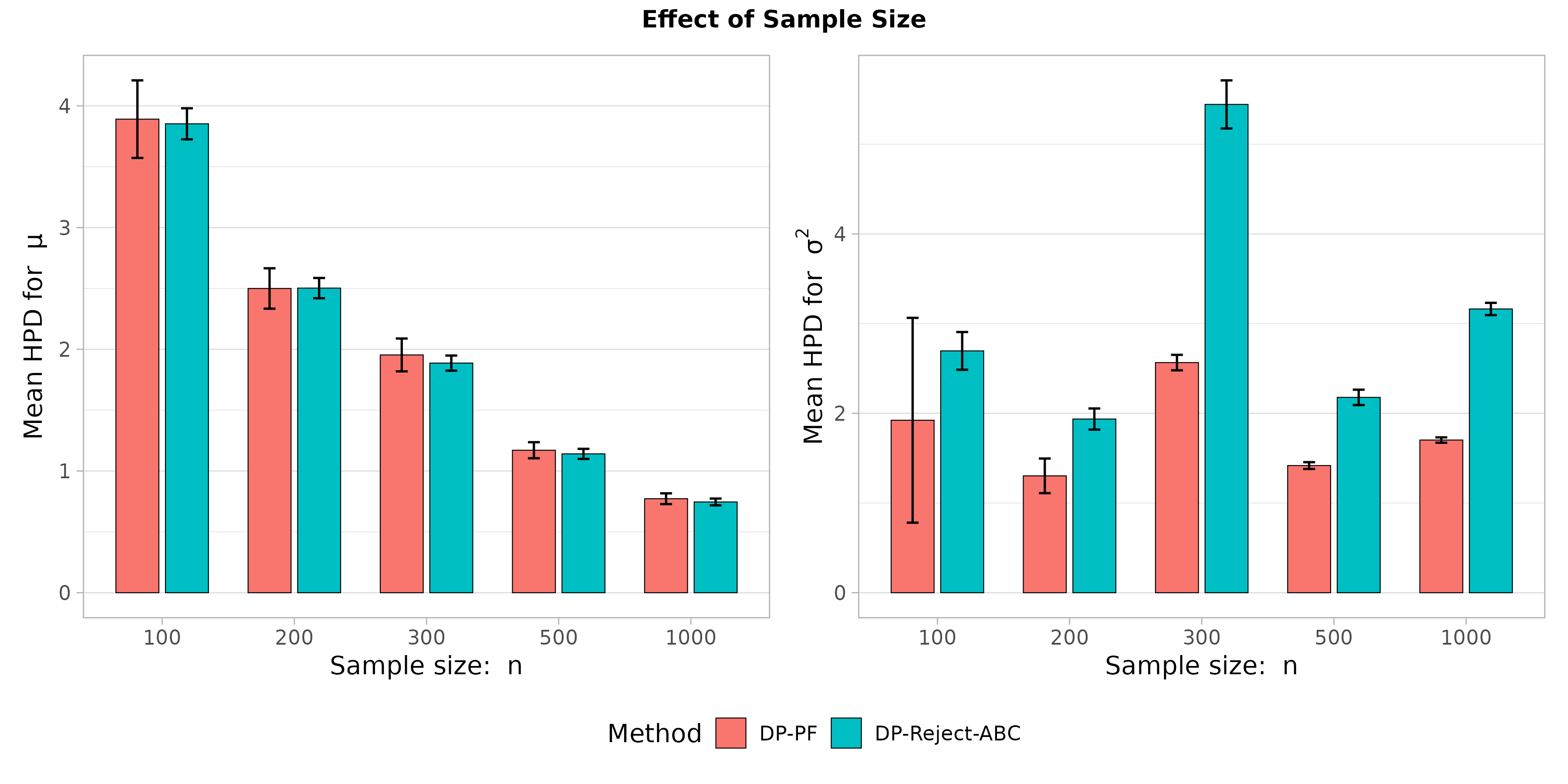}
    \caption{Comparison of the averaged lengths of the $90\%$ HPDs from DP-PF with the MCSEs of the posterior distributions of $\mu$ and $\sigma^2$ given $\sdp$, under different values of sample sizes.}
    \label{fig: HPD_sample_size}
\end{figure}

Figure~\ref{fig: HPD_sample_size} shows that the lengths of the $90\%$ HPDs for $\mu$ obtained by DP-PF and DP-Reject-ABC are nearly indistinguishable relative to their MCSEs. In contrast, DP-PF appears to produce slightly shorter HPDs for $\sigma^2$ than DP-Reject-ABC. Aside from the case $n=300$, where both methods exhibit an unexpected increase in HPD length, DP-PF appears to systematically underestimate the HPD length for $\sigma^2$. We expect this discrepancy to diminish with a larger number of particles, such as $N=2000$ or higher.

\subsection{Linear Regression: Conjugate prior} \label{subsec: linear regression with conjugate priors}

The data model $f$ is the same as the non-conjugate prior but with conjugate prior assumptions $\beta \mid \tau \sim N_{2}(m, \tau^{-1} V^{-1}), \tau \sim \text{Gamma}\left( \frac{a}{2}, \frac{b}{2} \right), \mu \sim N(\theta, \Sigma),$ and $\phi \sim \chi^2 (d)$. The privatized statistics $\sdp$ are the same as in the non-conjugate prior setting.
 
\emph{Simulation setup.} With bounds $L=-5, U=5$, true values of the parameters are set as $\beta^* = (0, 2)^\top$, $\tau^* = 1$, $\mu^* = 1$, and $\phi^* = 1$. $\epsilon\in\{0.5,1,2\}$ and $n\in\{300,500\}$. The hyper-parameters are set as $m = (0,0)^\top$, $V = I_2$, $a=2$, $b=2$, $\theta = (0)^\top$, $\Sigma = 1$, and $d=2$.

\begin{table}[t]
\centering
\resizebox{0.95\textwidth}{!}{%
\begin{tabular}{cc cc cc cc}
\toprule
&
&
\multicolumn{2}{c}{DP-PF} &
\multicolumn{2}{c}{DP-DA-MCMC} &
\multicolumn{2}{c}{DP-Reject-ABC} \\
\cmidrule(lr){3-4}
\cmidrule(lr){5-6}
\cmidrule(lr){7-8}
$\epsilon$ & Sample Size
& $\ex(\beta_1 \mid \sdp)$ & Sec./ESS
& $\ex(\beta_1 \mid \sdp)$ & Sec./ESS
& $\ex(\beta_1 \mid \sdp)$ & Sec./ESS \\
\midrule

\multirow{2}{*}{0.50}
& 300 & 1.906 (0.0133) & 0.296 & 1.679 (0.0223) & 24.456 & 1.842 (0.0030) & 0.150 \\
& 500 & 1.726 (0.0087) & 0.868 & 1.618 (0.0109) & 45.503 & 1.648 (0.0024) & 4.959 \\

\midrule

\multirow{2}{*}{1}
& 300 & 1.900 (0.0092) & 0.268 & 1.751 (0.0126) & 17.024 & 1.895 (0.0020) & 1.326 \\
& 500 & 1.845 (0.0083) & 2.796 & 1.762 (0.0071) & 36.384 & 1.791 (0.0015) & 64.206 \\

\midrule

\multirow{2}{*}{2}
& 300 & 1.951 (0.0066) & 0.591 & 1.896 (0.0095) & 12.984 & 1.957 (0.0013) & 17.053 \\
& 500 & 1.918 (0.0070) & 11.657 & 1.876 (0.0050) & 28.439 & $-$ & $-$ \\

\bottomrule
\end{tabular}%
}
\caption{Comparison of DP-PF, DP-DA-MCMC, and DP-Reject-ABC based on the averaged posterior mean of the regression coefficient $\beta_1$ and seconds per effective sample size (Sec./ESS) over 1000 simulations for linear regression under conjugate priors, with the values in parentheses representing the MCSEs for the posterior mean estimates. ($^\dagger$ Parallelized over particles.)}
\label{table:type2_conjugate}
\end{table}

\begin{table}[t]
\centering
\resizebox{0.65\textwidth}{!}{%
\begin{tabular}{cc ccc}
\toprule
&
&
DP-PF & DP-DA-MCMC & DP-Reject-ABC \\
\cmidrule(lr){3-3}
\cmidrule(lr){4-4}
\cmidrule(lr){5-5}
$\epsilon$ & Sample Size
& 90\% HPD & 90\% HPD & 90\% HPD \\
\midrule
\multirow{2}{*}{0.5}
& 300 & 1.897 (0.0099) & 2.825 (0.0386) & 2.664 (0.0092) \\
& 500 & 1.507 (0.0076) & 1.927 (0.0243) & 2.023 (0.0068) \\
\midrule
\multirow{2}{*}{1.0}
& 300 & 1.635 (0.0073) & 2.178 (0.0243) & 1.905 (0.0060) \\
& 500 & 1.192 (0.0054) & 1.422 (0.0139) & 1.401 (0.0044) \\
\midrule
\multirow{2}{*}{2.0}
& 300 & 1.275 (0.0058) & 1.672 (0.0173) & 1.349 (0.0040) \\
& 500 & 0.844 (0.0036) & 1.038 (0.0092) & $-$ \\
\bottomrule
\end{tabular}%
}
\caption{Comparison of DP-PF, DP-DA-MCMC, and DP-Reject-ABC based on the averaged 90\% HPD widths of $\beta_1$ over 1000 simulations for linear regression under conjugate priors, with values in parentheses representing MCSEs.}
\label{table:type2_conjugate_hpd}
\end{table}

Table~\ref{table:type2_conjugate} along with Table~\ref{table:type2_nonconjugate} show that DP-PF consistently delivers strong accuracy across different sample sizes, privacy levels, and prior settings. Although DP-PF does not gain as much computational speed from conjugate priors as DP-DA-MCMC, it offers a key advantage that its implementation remains consistent and straightforward regardless of the prior.  In contrast, while DP-DA-MCMC benefits from a simplified structure when conjugate priors are available, it becomes notably more complex to implement and slower to run in non-conjugate settings. This highlights a key strength of DP-PF that its simplicity and flexibility make it easier to apply across a wide range of models.

\subsection{Linear Regression: on \texorpdfstring{$\epsilon_t$}{epsilon	extunderscore t} Schedule}

To demonstrate the performance under different $\epsilon_t$ schedules, we implemented an exponential schedule in addition to the linear one using linear regression under the conjugate and non-conjugate priors stated in Section~\ref{subsec: simulation on consistency} and Appendix~\ref{subsec: linear regression with conjugate priors}. Under this schedule, the privacy budget increases rapidly during the early iterations and gradually approaches the total budget $\epsilon$. Specifically,
\begin{align*}
\epsilon_t = \epsilon \left(1 - 0.9 e^{-0.3(t-1)}\right), \qquad t = 1,\ldots,T-1,    
\end{align*}
and $\epsilon_T = \epsilon$.

\begin{table}[t]
\centering
\resizebox{0.95\textwidth}{!}{%
\begin{tabular}{cc ccc ccc}
\toprule
&
&
\multicolumn{3}{c}{DP-PF (linear $\epsilon_t$ schedule)} &
\multicolumn{3}{c}{DP-PF (exponential $\epsilon_t$ schedule)} \\
\cmidrule(lr){3-5}
\cmidrule(lr){6-8}
$\epsilon$ & Sample Size
& $\ex(\beta_1 \mid \sdp)$ & $\var(\beta_1 \mid \sdp)$ & Sec./ESS$^\dagger$
& $\ex(\beta_1 \mid \sdp)$ & $\var(\beta_1 \mid \sdp)$ & Sec./ESS$^\dagger$ \\
\midrule
\multirow{2}{*}{0.5}
& 300 & 1.973 (0.0128)    & 0.523 (0.0133) & 0.293 & 1.970 (0.0136)    & 0.554 (0.0167) & 0.344 \\
& 500 & 1.772 (0.0079)    & 0.295 (0.0065) & 0.638 & 1.778 (0.0084)    & 0.293 (0.0066) & 1.010 \\
\midrule
\multirow{2}{*}{1.0}
& 300 & 1.983 (0.0085)    & 0.314 (0.0062) & 0.250 & 1.987 (0.0094)    & 0.325 (0.0064) & 0.349 \\
& 500 & 1.889 (0.0070)    & 0.178 (0.0029) & 1.696 & 1.888 (0.0071)    & 0.182 (0.0056) & 3.089 \\
\midrule
\multirow{2}{*}{2.0}
& 300 & 2.048 (0.0076)    & 0.187 (0.0051) & 0.634 & 2.052 (0.0079)    & 0.185 (0.0073) & 1.078 \\
& 500 & 1.980 (0.0061)    & 0.095 (0.0038) & 7.481 & 1.980 (0.0064)    & 0.099 (0.0027) & 14.021 \\
\bottomrule
\end{tabular}%
}
\caption{Comparison of DP-PF with linear and exponential $\epsilon_t$ schedules based on the averaged posterior means and variances of $\beta_1$ over 1000 simulations for a linear regression model under a non-conjugate prior, and the seconds spent per ESS, with the values in the parentheses representing MCSEs for the posterior mean and variance estimates. ($^\dagger$Parallelized over particles.)}
\label{table:nonconj_exp}
\end{table}

\begin{table}[t]
\centering
\resizebox{0.95\textwidth}{!}{%
\begin{tabular}{cc ccc ccc}
\toprule
&
&
\multicolumn{3}{c}{DP-PF (linear $\epsilon_t$ schedule)} &
\multicolumn{3}{c}{DP-PF (exponential $\epsilon_t$ schedule)} \\
\cmidrule(lr){3-5}
\cmidrule(lr){6-8}
$\epsilon$ & Sample Size
& $\ex(\beta_1 \mid \sdp)$ & $\var(\beta_1 \mid \sdp)$ & Sec./ESS$^\dagger$
& $\ex(\beta_1 \mid \sdp)$ & $\var(\beta_1 \mid \sdp)$ & Sec./ESS$^\dagger$ \\
\midrule
\multirow{2}{*}{0.5}
& 300 & 1.906 (0.0133)    & 0.581 (0.0127) & 0.296 & 1.920 (0.0145)    & 0.615 (0.0175) & 0.421 \\
& 500 & 1.726 (0.0087)    & 0.373 (0.0067) & 0.868 & 1.713 (0.0097)    & 0.373 (0.0077) & 1.380 \\
\midrule
\multirow{2}{*}{1.0}
& 300 & 1.900 (0.0092)    & 0.310 (0.0073) & 0.268 & 1.901 (0.0094)    & 0.306 (0.0062) & 0.380 \\
& 500 & 1.845 (0.0083)    & 0.207 (0.0118) & 2.796 & 1.841 (0.0074)    & 0.185 (0.0040) & 4.202 \\
\midrule
\multirow{2}{*}{2.0}
& 300 & 1.951 (0.0066)    & 0.161 (0.0046) & 0.591 & 1.952 (0.0062)    & 0.157 (0.0033) & 0.990 \\
& 500 & 1.918 (0.0070)    & 0.104 (0.0034) & 11.657 & 1.937 (0.0071)    & 0.111 (0.0047) & 22.512 \\
\bottomrule
\end{tabular}%
}
\caption{Comparison of DP-PF with linear and exponential $\epsilon_t$ schedules based on the averaged posterior means and variances of $\beta_1$ over 1,000 simulations for a linear regression model under a conjugate prior, and the seconds spent per ESS, with the values in the parentheses representing MCSEs for the posterior mean and variance estimates. ($^\dagger$Parallelized over particles.)}
\label{table:conj_exp}
\end{table}

From Tables~\ref{table:nonconj_exp} and \ref{table:conj_exp}, we observe that the posterior estimates and their MCSEs are similar under the two schedules. Although the linear schedule is somewhat more computationally efficient in terms of seconds per ESS, the exponential schedule achieves comparable overall performance.

\section{Details of Real Data Analysis} \label{appendix: real data analysis}

For implementing Algorithm~\ref{alg: DP-PF} with K-norm mechanisms, we need the following proposition.

\begin{proposition}
    If the K-norm mechanism is used as the privacy mechanism in Algorithm~\ref{alg: DP-PF} for $T_z = (\sum z_i, \sum z_i^2)$, then the acceptance probability for the sample $z$ is
    \begin{align*}
        r_t (z) = \exp\left( \frac{-\epsilon_t}{\Delta} \| T_z - Z_{dp} \|_K \right)
    \end{align*}
    at iteration $t$.
\end{proposition}

\begin{proof}
The K-Norm mechanism \citep{hardt2010geometry} adds a random variable $V$ to the summary statistics with density 
$$
f_v(v) = \frac{\exp(\frac{-\epsilon}{\Delta} \|v\|_K)}{\Gamma(d+1) \lambda(\frac{\Delta}{\epsilon} K)},
$$
which satisfies $\epsilon$-DP, where $\Gamma(\cdot)$ is the Gamma function and $\lambda(\cdot)$ is the Lebesgue volume function. Given fixed dimensions $d$ and $K$, since both of the functions remain constant, the acceptance rate follows directly from calculating
    \begin{align*}
        r_t (z)
        = \frac{m_{\epsilon_t} (Z_{dp} | z)}{\sup_z m_{\epsilon_t} (Z_{dp}  | z)} 
        = \exp\left( \frac{-\epsilon_t}{\Delta} \| T_z - Z_{dp} \|_K - \frac{-\epsilon_t}{\Delta} \| 0 \|_K \right) 
        = \exp\left( \frac{-\epsilon_t}{\Delta} \| T_z - Z_{dp} \|_K\right) .
    \end{align*}
\end{proof}

In general, sampling from a K-norm mechanism can be challenging as it requires a precise characterization of the sensitivity space and relies on rejection sampling to give an exact uniform sampling from the sensitivity space \citep{awan2021structure}. For the $\ell_\infty$-norm used in our data analysis, Algorithm~\ref{alg: sampling from ell_infty} provides a simple way to sample from random vectors whose density is proportional to their $\ell_\infty$-norm.

\begin{algorithm}[t]
\caption{Sampling $\propto \exp(-c\| V \|_\infty)$: \citet{steinke2016between}} \label{alg: sampling from ell_infty}
\begin{algorithmic}[1]
\Require $c$ and dimension $d$
    \State Set $U_j$ for $j=1,\ldots,d$
    \State Draw $r$ from $\Gamma(\alpha = d+1, \beta = c)$
    \State Set $V = r \cdot (U_1,\ldots,U_d)^\top$
\Ensure $V$
\end{algorithmic}
\end{algorithm}

\begin{algorithm}[t]
\caption{$\ell_\infty$ Objective Perturbation: \citet{awan2021structure}} \label{alg: ell_infty objective perturbation}
\begin{algorithmic}[1]
\Require $X\in\mathcal{X}^n,\varepsilon>0$, a convex set $\Theta\subset\mathbb{R}^d$, a convex function $r:\Theta\to\mathbb{R}$, a convex loss $\hat{\mathcal{L}}(\theta;X)=\frac1n\sum_{i=1}^n\ell(\theta;x_i)$ defined on $\Theta$ such that $\nabla^2\ell(\theta;x)$ is continuous in $\theta$ and $x,\Delta>0$ such that $\sup_{x,x^\prime\in\mathcal{X}}\sup_{\theta\in\Theta}\|\nabla\ell(\theta;x) -\nabla\ell(\theta;x^{\prime})\|_\infty\leq\Delta,\lambda>0$ is an upper bound on the eigenvalues of $\nabla^2\ell(\theta;x)$ for all $\theta\in\Theta$ and $x\in\mathcal{X}$, and $q\in (0,1).$ 
    \State Set $\gamma = \frac{\lambda}{\exp(\varepsilon(1-q))-1}$
    \State sample $V$ from density $\propto \exp(\frac{-\varepsilon q}{\Delta} \|V\|_\infty)$ from Algorithm~\ref{alg: sampling from ell_infty}
    \State Compute $\theta_{dp} = \arg\min_{\theta} \frac{1}{n} \sum_{i=1}^n l(\theta, x_i) + \frac{1}{n} r(\theta) + {\frac{\gamma}{2n} \theta^\top \theta} + {\frac{1}{n} V^\top \theta}$.
\Ensure $\theta_{dp}$
\end{algorithmic}
\end{algorithm}

To implement Algorithm~\ref{alg: DP-PF} with the objective perturbation mechanism described in Algorithm~\ref{alg: ell_infty objective perturbation} under a logistic model, we need the following proposition.

\begin{proposition}
    Under a logistic model, if the objective perturbation is used as the privacy mechanism in Algorithm~\ref{alg: ell_infty objective perturbation} for the regression coefficient $\theta = (\beta_0,\beta_1,\ldots,\beta_{d-1})$, then the acceptance probability for the data set $D=(X,Y)$ is
    \begin{align*}
        r_t(D) = \frac{ \exp{(\frac{-\epsilon_t}{\Delta_K} \| V(\theta_{dp};D) \|_K}) }{(\gamma_t + \lambda_1)(\gamma_t + \lambda_2)\cdots(\gamma_t + \lambda_d)}  \gamma_t^d
    \end{align*}
    at iteration $t$, where $\gamma_t = \frac{d/4}{\exp(\epsilon_t (1-q_t))-1}$ and $\lambda_1 \geq \lambda_2 \geq \cdots \geq \lambda_d$ are the eigenvalues of $ \sum_{i=1}^n \left( \frac{\exp{(\theta_{dp}^\top X_i)}}{(1 + \exp{(\theta_{dp}^\top X_i))^2}}  \right) X_i X_i^{\top}$.
\end{proposition}

\begin{proof}
For iteration $t$, let $\gamma_t = \frac{d/4}{\exp(\epsilon_t (1-q_t))-1}$. Consider
$$
\theta_{dp} = \arg\min_\theta \Hat{\mathcal{L}}(\theta;D) + \frac{\gamma_t}{2n}\theta^{\top}\theta + \frac{1}{n} V^\top \theta,
$$
where $V$ is drawn from $f(V) \propto \exp(\frac{-\epsilon_t }{\Delta_K} \| V \|_K )$ and $0<q_t<1$. Let $\theta_{dp} = a$, then we have the identity $\nabla_\theta \Hat{\mathcal{L}}(a;X) + \frac{\gamma_t}{n} a + \frac{1}{n} V = 0$. Thus, we know
\begin{align*}
    V(a;D) = - \left( \gamma_t a  + n \nabla_\theta \Hat{\mathcal{L}}(a;D) \right) \\
    \nabla_\theta V(a;D) = - \left( \gamma_t I  + n \nabla_\theta^2 \Hat{\mathcal{L}}(a;D) \right).
\end{align*}

Then, the pdf of $\theta_{dp} = a$ given $X$ can be expressed with change of variable
\begin{align*}
    \frac{f_V(V(a;D))}{|\det \nabla V(a;D)|} 
    &= \frac{f_V \left( -( \gamma_t a  + n \nabla_\theta \Hat{\mathcal{L}}(a;D) )\right)}{\det \left( \gamma_t I  + n \nabla_\theta^2 \Hat{\mathcal{L}}(a;D) \right)} \\
    &= \frac{\frac{1}{\Gamma(d+1) \lambda(\frac{\Delta}{\epsilon_t } K)} \exp{(\frac{-\epsilon_t }{\Delta} \| - \left( \gamma_t a  + \sum_{i=1}^n (\frac{\exp{(a^\top X_i)}}{1+\exp{(a^\top X_i)}} - Y_i) X_i \right) \|_K}) }{(\gamma_t + \lambda_1)(\gamma_t + \lambda_2)\cdots(\gamma_t + \lambda_d)}
\end{align*}
because the Hessian of $\Hat{\mathcal{L}}$ can be expressed as
\begin{align*}
    \nabla_\theta^2 \Hat{\mathcal{L}}(a;D)
    &= \frac{1}{n} \sum_{i=1}^n \left( \frac{\exp{(a^\top X_i)}}{(1 + \exp{(a^\top X_i))^2}}  \right) X_i X_i^{\top},
\end{align*}
where the eigenvalues of $n \nabla_\theta^2 \Hat{\mathcal{L}}(a;D)$ for the synthetic data $D$ are denoted as $\lambda_1 \geq \lambda_2 \geq \cdots \geq \lambda_d$. Thus,
\begin{align*}
    m_{\epsilon_t}  (\theta_{dp} = a \mid D) &= \frac{ C_{\epsilon_t} \exp{(\frac{-\epsilon_t}{\Delta} \| V(a;D) \|_K}) }{(\gamma_t + \lambda_1)(\gamma_t + \lambda_2)\cdots(\gamma_t + \lambda_d)},
\end{align*}
where $C_{\epsilon_t} = \frac{1}{\Gamma(d+1) \lambda(\frac{\Delta}{\epsilon_t} K)}$. Note that $\sup_X m_{\epsilon_t} (\theta_{dp} = a \mid D) = \frac{C_{\epsilon_t}}{\gamma_t^d}$. Therefore, the acceptance rate of our differentially private particle filtering algorithm at iteration $t$ is
$$
r_t (D) 
= \frac{m_{\epsilon_t} (\theta_{dp} = a \mid D)}{\sup_X m_{\epsilon_t} (\theta_{dp} = a \mid D)} 
= \frac{ \exp{(\frac{-\epsilon_t}{\Delta} \| V(a;D) \|_K}) }{(\gamma_t + \lambda_1)(\gamma_t + \lambda_2)\cdots(\gamma_t + \lambda_d)}  \gamma_t^d.
 $$
\end{proof}

\section{Proofs and Technical Details} \label{appendix: proofs}

Since our propagation involves perturbation and rejection, the whole propagation transition has density 
\begin{align*}
    h_t (\theta, x)  = \frac{ \eta_t(\theta) f(x \mid \theta) m_{\epsilon_t}(\sdp \mid x)}{\int \int \eta_t(\theta) f(x \mid \theta) m_{\epsilon_t}(\sdp \mid x) \ dx \ d\theta}
\end{align*}
 with its denominator denoted as $c_{\eta_t} (\sdp)$. We prove that the accepted particles after the whole propagation are drawn from $h_t$ in Proposition~\ref{prop: acceptance distribution}. In other words, the accepted particles are derived from conditioning on the successful events after perturbation. That is,
 \begin{align*}
    h_t(\cdot, x) = \int \pi_{\epsilon_{t-1}} (\tilde{\theta}) K_t(\tilde{\theta}, \cdot \mid A_t) d\tilde{\theta},
 \end{align*}
 where $A_t = \{ (\theta,x) : \pi_0 (\theta) > 0, r_t(x) > U \}$ is the event of success. The importance weight becomes
\begin{align*}
    w_t = \frac{\pi_{\epsilon_t}(\theta, x \mid \sdp)}{h_t(\theta, x)},
\end{align*}
where 
\begin{align*}
    \pi_{\epsilon_t} (\theta, x \mid \sdp) = \frac{\pi_0 (\theta) f(x \mid \theta) m_{\epsilon_t}(\sdp \mid x)}{\int \int \pi_0 (\theta) f(x \mid \theta) m_{\epsilon_t}(\sdp \mid x) \ dx \ d\theta}
\end{align*}
with its denominator denoted as $c_{\epsilon_t}(\sdp)$. Often, both of the normalizing constants $c_{\eta_t} (\sdp)$ and $c_{\epsilon_t}(\sdp)$ are intractable. Nevertheless, this can be addressed by using a consistent estimate of the fraction of the two normalizing constants (Lemma~\ref{lemma: convergence of posterior mean}).

\subsection{Properties of DP-PF (Algorithm~\ref{alg: DP-PF})} \label{appendix: properties of DP-PF}

In this section, we present two useful properties that are a direct consequence of the way Algorithm~\ref{alg: DP-PF} is constructed.
    
 \begin{proposition} \label{prop: acceptance distribution}
     The accepted samples $(\theta, x)^{(i,t)}$ in algorithm~\ref{alg: DP-PF} follow distribution $h_t(\theta, x)$. Furthermore, let $Y^{(i,t)}$ be the trial times ($\geq 1$) until the first acceptance. Then, 
     \begin{itemize}
         \item $Y^{(i,t)}$ is independent of the accepted samples $(\theta, x)^{(i,t)}$, and
         \item $Y^{(i,t)}$ follows $Geom \left( p = \frac{c_{\eta_t}(\sdp)}{\sup_{x} m_{\epsilon_t}(\sdp \mid x)} \right)$.
     \end{itemize}
 \end{proposition}
 
\begin{proof}
Define $I_t = 1$, if $(\theta, x) \in A_t$; $I_t = 0$, otherwise. Since 
\begin{align*}
    P(I_t = 1 \mid \Theta_t=\theta, X=x) = 
    \frac{m_{\epsilon_t}(\sdp \mid x)}{\sup_{x} m_{\epsilon_t}(\sdp \mid x)},
\end{align*}
the probability of acceptance is
    \begin{align*}
        P(I_t = 1) &= \int \int P(I_t = 1 \mid \Theta_t=\theta, X=x) \ \eta(\theta) f(x \mid \theta) \ d\theta dx \\
        &= \frac{1}{\sup_{x} m_{\epsilon_t}(\sdp \mid x)} \int \int \eta(\theta) f(x \mid \theta) m_{\epsilon_t}(\sdp \mid x) \ d \theta dx \\
        &= \frac{c_{\eta_t}(\sdp)}{\sup_{x} m_{\epsilon_t}(\sdp \mid x)}.
    \end{align*}
    The distribution of the accepted samples is 
    \begin{align*}
        P(\Theta_t = \theta, X = x \mid I_t = 1) &= \frac{P(I_t = 1 \mid \Theta_t = \theta, X = x) P(\Theta_t = \theta, X = x)}{P(I_t = 1)} \\
        &= \frac{\eta_t(\theta) f(x \mid \theta) m_{\epsilon_t}(\sdp \mid x)}{c_{\eta_t}(\sdp)} \\
        &= h_t(\theta, x).
    \end{align*}
    Thus, the density of the accepted samples is $h_t(\theta, x)$. Now, consider 
    \begin{align*}
        & P(Y = s, \Theta_t = \theta, X = x) \\
        =& \left( 1 - \frac{c_{\eta_t}(\sdp)}{\sup_{x} m_{\epsilon_t}(\sdp \mid x)} \right)^{s-1}  \eta_t(\theta) f(x \mid \theta) \frac{m_{\epsilon_t}(\sdp \mid x)}{\sup_{x} m_{\epsilon_t}(\sdp \mid x)}\\
        =& \left[ \left( 1 - \frac{c_{\eta_t}(\sdp)}{\sup_{x} m_{\epsilon_t}(\sdp \mid x)} \right)^{s-1}  \frac{c_{\eta_t}(\sdp)}{\sup_{x} m_{\epsilon_t}(\sdp \mid x)} \right] \cdot \frac{\eta_t(\theta) f(x \mid \theta) m_{\epsilon_t}(\sdp \mid x)}{c_{\eta_t}(\sdp)}\\
        =& P(G=s) \cdot h_t(\theta, x).
    \end{align*}
    The conclusion follows.
\end{proof}

\evidence*

\begin{proof}
For the first result, observe that for any $t \ge 1$,
\begin{align*}
    c_{\epsilon_t}(\sdp)
    =
    \int \kappa_t(\theta,x) \, d\theta dx
    =
    \int \kappa_{t-1}(\theta,x)
    \frac{\kappa_t(\theta,x)}
         {\kappa_{t-1}(\theta,x)}
    \, d\theta dx.
\end{align*}
Dividing both sides by $c_{\epsilon_{t-1}}(\sdp)$ yields
\begin{align*}
    \rho_t \;=\; \frac{c_{\epsilon_t}(\sdp)}{c_{\epsilon_{t-1}}(\sdp)}
    &= \int \frac{\kappa_{t-1}(\theta,x)}{c_{\epsilon_{t-1}}(\sdp)} \frac{\kappa_t(\theta,x)}   {\kappa_{t-1}(\theta,x)} \, d\theta dx \\
    &= \mathbb{E}_{\pi_{\epsilon_{t-1}}(\theta,x \mid \sdp)} \left[ \frac{\kappa_t(\theta,x)}       {\kappa_{t-1}(\theta,x)} \right] \\
    &= \mathbb{E}_{\pi_{\epsilon_{t-1}}(\theta,x \mid \sdp)} \left[ \frac{m_{\epsilon_t}(\sdp \mid x)}{m_{\epsilon_{t-1}}(\sdp \mid x)} \right],
\end{align*}
where the first identity follows from $\kappa_t(\theta,x) = \pi_{\epsilon_{t}}(\theta,x \mid \sdp)c_{\epsilon_{t}}(\sdp)$, and the second identity 
\begin{align*}
    \frac{\kappa_t(\theta,x)}{\kappa_{t-1}(\theta,x)}
    = \frac{\pi_0(\theta) f(x \mid \theta) m_{\epsilon_t}(\sdp \mid x)}{\pi_0(\theta) f(x \mid \theta) m_{\epsilon_{t-1}}(\sdp \mid x) }
    = \frac{m_{\epsilon_t}(\sdp \mid x)}{m_{\epsilon_{t-1}}(\sdp \mid x)}
\end{align*}
by the design of Algorithm~\ref{alg: DP-PF}. The factorization $c_{\epsilon_T}(\sdp) = \prod_{t=1}^T \rho_t$ follows by telescoping.

For the second result, define $\varphi_t(\theta, x) = m_{\epsilon_t}(\sdp \mid x) / m_{\epsilon_{t-1}}(\sdp \mid x)$. By the integrability assumption of the theorem, $\ex_{h_t}|w_t \varphi_t| < \infty$, so Theorem~\ref{thm: posterior mean convergence} applies with $\varphi = \varphi_t$, giving
\begin{align*}
    \hat{\rho}_t
    \;=\;
    \sum_{i=1}^N \textup{\textbf{w}}^{(i,t-1)} \, \varphi_t\!\left(\theta^{(i,t-1)}, x^{(i,t-1)}\right)
    \;\overset{a.s.}{\longrightarrow}\;
    \ex_{\pi_{\epsilon_{t-1}}(\theta,x \mid \sdp)}[\varphi_t(\theta, x)]
    \;=\; \rho_t
\end{align*}
as $N \to \infty$. Since $\widehat{c_{\epsilon_T}(\sdp)} = \prod_{t=1}^T \hat{\rho}_t$ is a continuous function of $(\hat{\rho}_1, \ldots, \hat{\rho}_T)$, the continuous mapping theorem gives $\widehat{c_{\epsilon_T}(\sdp)} \overset{a.s.}{\longrightarrow} c_{\epsilon_T}(\sdp)$.
\end{proof}

\subsection{On the Consistency and CLT of the DP-PF Estimator}

In order to carefully define the measurability across iterations for the consistency and CLT proof. We use the following definition:

\begin{definition} [Recursive Measurability \citealp{chopin2004}] \label{def: recursive measurability}
    ${\mathcal{F}_t}^p (d)$ is defined to be the set of measurable function $\varphi: \Theta_t \times \mathcal{X}^n \to \mathbb{R}^d$ such that for some $p$,
    \begin{align*}
        \ex_{h_t} |w_t \varphi|^p < \infty,
    \end{align*}
    in which the function $\theta_{t-1} \mapsto \ex_{K_t(\theta_{t-1}, \cdot \mid A_t)} [w_t(\cdot) \ \varphi(\cdot)]$ lies in ${\mathcal{F}_{t-1}}^p (1)$. For convenience, if $d=1$, the $(1)$ will be omitted, and if $p=2^\dag$, then there exists some $\zeta >0$, such that $\varphi \in {\mathcal{F}_t}^{2+\zeta} (d)$.
\end{definition}

From section~\ref{sec: particle filters for DP}, we know that after the whole propagation transition (\emph{Step 2. $\&$ 2.5.}) in each $t$, we are able to sample from the proposal distribution $h_t (\theta, x)$. Thus, we use importance sampling to derive the importance weights $w_t = \frac{\pi_{\epsilon_t} (\theta, x \mid \sdp)}{h_t (\theta, x)}$, such that
    \begin{align} \label{eq: importance sampling dp}
         \ex_t [\varphi] = \ex_{h_t(\theta, x)} [w_t \varphi].
    \end{align}
    
While $w_t$ satisfies \eqref{eq: importance sampling dp}, the normalizing constants of the two distributions are often intractable. Alternatively, we tackle
\begin{align*}
    \Tilde{w}_t = \frac{\pi_0(\theta)}{\eta_t(\theta)},
\end{align*}
such that $\Tilde{w}_t = w_t \gamma_t$ with $\gamma_t \coloneqq \left( \frac{c_{\eta_t}(\sdp)}{c_{\epsilon_t}(\sdp)} \right)^{-1}$ being constant at each iteration.

In exchange for not being able to quantify the normalizing constants, we need the following lemma to set up for consistency.

 \begin{lemma} \label{lemma: convergence of posterior mean}
 Let $\varphi: \Theta_t \times \mathcal{X}^n \to \mathbb{R}^d$, such that $\ex_{h_t} |w_t \varphi_j| < \infty$ for all $j=1,\ldots,d$,
  \begin{align} \label{eq: convergence of posterior mean}
   \frac{1}{N} \sum_{i=1}^N \Tilde{w}^{(i,t)}  \varphi \left( (\theta, x)^{(i,t)} \right) \overset{a.s.}{\to} \gamma_t \ex_t [\varphi]
  \end{align}
  as $N \to \infty$.  In the case $\varphi \equiv 1$, we obtain $\frac{1}{N} \sum_{i=1}^N \Tilde{w}^{(i,t)} \overset{a.s.}{\to} \gamma_t$.
 \end{lemma}

\begin{proof}
    The argument can be applied coordinate-wise, so we prove the result for $d=1$. 
    For any $\varphi$ satisfying $\ex_{h_t} |w_t \varphi| < \infty$,
    $$
    \ex_{h_t(\theta, x)} [\Tilde{w}_t \varphi] =  \gamma_t \ex_t [\varphi].
    $$
   Now, let $Y_i = \Tilde{w}^{(i,t)} \varphi \left( (\theta, x)^{(i,t)} \right)$ for all $i=1,..,N$. Since $Y_1,\ldots,Y_N$ are i.i.d., the conclusion follows by the strong law of large numbers (See Section 2.4 in \citet{durrett2019probability}).
\end{proof}

We use Lemma~\ref{lemma: convergence of posterior mean} to prove the strong consistency of our estimator $\ex_t[\varphi]$.

\strongConsistency*

\begin{proof} 
     The argument can be applied coordinate-wise, so we prove the results for $d=1$. Let $f(x) = \frac{1}{x}$, which is continuous on $x>0$. Since $\frac{1}{N} \sum_{i=1}^N \Tilde{w}^{(i,t)} \overset{a.s.}{\to} \gamma_t$, 
    \begin{align*}
        f \left( \frac{1}{N} \sum_{i=1}^N \Tilde{w}^{(i,t)} \right) = \frac{1}{\frac{1}{N} \sum_{i=1}^N \Tilde{w}^{(i,t)}} \overset{a.s.}{\to} \frac{1}{\gamma_t}.
    \end{align*}
    by the continuous mapping theorem (Theorem 3.2.10 in \citet{durrett2019probability}). With $\eqref{eq: convergence of posterior mean}$, the conclusion follows by the continuous mapping theorem for a bivariate function.
\end{proof}

Another crucial result is the central limit theorem for our estimator $\ex_t[\varphi]$. It requires a higher-order integrability assumption.  

\clt*

\begin{proof}
    In addition to recursive measurability, we also assume the unit function $\theta_t \mapsto 1$ belongs to ${\mathcal{F}_T}^{2^\dag}$. This, for example, ensures $\frac{1}{N} \sum_{i=1}^N \Tilde{w}^{(i,t)} \overset{a.s.}{\to} \gamma_t$ and sets the ground for CLT.
    
    Define the following quantities for the asymptotic variances in the CLT. Let $\varphi: \Theta_0 \times \mathcal{X}^n \to \mathbb{R}^d$, and $\widetilde{V}_0(\varphi) = \var_{\pi_0}(\varphi)$, which is a typical asymptotic variance of a CLT where the sample mean converges to its population mean of $\pi_0$. Then, by induction, $\forall \varphi: \Theta_t \times \mathcal{X}^n \to \mathbb{R}^d$, let
    \begin{align}
        \widecheck{V}_t (\varphi) &\coloneqq \widetilde{V}_{t-1} [\ex_{K_t(\cdot, \cdot \mid A_t)}(\varphi)] + \ex_{\pi_{\epsilon_{t-1}}}[\var_{K_t (\cdot, \cdot \mid A_t)}(\varphi)], \quad &t > 0, \ (propagation) \label{eq: variance of reweighting} \\
        V_t (\varphi) &\coloneqq \widecheck{V}_t \left[ w_t \left( \varphi - \ex_{\pi_{\epsilon_{t}}} (\varphi) \right) \right], &t > 0, \ (reweighting) \nonumber\\
        \widetilde{V}_t (\varphi) &\coloneqq V_t(\varphi) + U_t (\varphi), &t > 0. \ \ (resampling) \nonumber
    \end{align}
    As stated by \citet{chopin2004}, the CLT proof of Theorem~\ref{thm: clt} is done by mathematical induction with three lemmas, each of which states a CLT for each step of the particle filtering. Essentially, the lemmas quantify the randomness incurred from doing resampling, propagating, and re-weighting. The inductive hypothesis is as follows, for any given $t>0$, we assume that $\forall \varphi \in {\mathcal{F}_{t-1}}^{2^\dag} (d)$,
    \begin{align} \label{the inductive hypothesis}
        \sqrt{N} \left[ \frac{1}{N} \sum_{j=1}^N \varphi \left( (\tilde{\theta},x)^{(j,t-1)} \right) - \ex_{\pi_{\epsilon_{t-1}}} (\varphi) \right] \overset{d}{\to} N(0, \widetilde{V}_{t-1}(\varphi)).
    \end{align}

\begin{lemma}[propagation: \citealp{chopin2004}] \label{lemma: propagation}
    Let $\psi: \Theta_t \to \mathbb{R}^d$, such that the function $\nu: \theta_{t-1} \mapsto \ex_{K_t(\theta_{t-1},\cdot \mid A_t)} [\psi(\cdot) - \ex_{h_t}(\psi)]$ lies in ${\mathcal{F}_{t-1}}^{2^\dag}$ and there exists $\zeta > 0$ such that $\ex_{h_t} | \psi |^{2+\zeta} < \infty$, then 
    \begin{align*}
        \sqrt{N} \left[ \frac{1}{N} \sum_{j=1}^N \psi( (\theta,x)^{(j,t)}) - \ex_{h_t} (\psi) \right] \overset{d}{\to} N(0, \widecheck{V}_{t-1}(\psi)).
    \end{align*}
    under the inductive hypothesis \eqref{the inductive hypothesis}.
\end{lemma}

Lemma~\ref{lemma: propagation} quantifies the MCSE from propagation. The asymptotic variance accounts for randomness from the previous target posterior $\pi_{\epsilon_{t-1}}$ and the propagation process. To adapt the proof to our DP-PF algorithm, we redefine the function $\nu$ slightly differently from its original proof, where $\nu$ represents a $(2+\delta)$-th central moment of $\psi$ with respect to its perturbation kernel. In our DP setting, after perturbation, particles are selected based on the success event $A_t = \{ \pi_0 (\cdot) > 0, r_t(x) > U \}$. Consequently, the propagation distribution becomes $K_t(\theta_{t-1},\cdot \mid A_t)$, incorporating both the kernel perturbation and particle selection to target the posteriors.

\begin{lemma}[reweighting: \citealp{chopin2004}] \label{lemma: reweighting}
    Let $\varphi \in {\mathcal{F}_{t}}^{2^\dag} (d)$ and the function $\theta_t \mapsto 1$ lies in ${\mathcal{F}_{t}}^{2^\dag}$, then
    \begin{align*}
        \sqrt{N} \left( \hat{\ex}_t[\varphi] - \ex_t[\varphi] \right) \overset{d}{\to} N(0, V_t(\varphi)),
    \end{align*}
    under the inductive hypothesis \eqref{the inductive hypothesis}.
\end{lemma}

Lemma~\ref{lemma: reweighting} quantifies the MCSE from reweighting. From \eqref{eq: variance of reweighting}, we can see that the asymptotic variance of reweighting is a rescale from that of propagation. At $t=T$, Algorithm~\ref{alg: DP-PF} stops at this step.

\begin{lemma}[resampling: \citealp{chopin2004}] \label{lemma: resampling}
    Let $\varphi \in {\mathcal{F}_{t}}^{2^\dag} (d)$ and the function $\theta_t \mapsto 1$ lies in ${\mathcal{F}_{t}}^{2^\dag}$, then
    \begin{align*}
        \sqrt{N} \left[ \frac{1}{N} \sum_{j=1}^N \varphi \left( (\tilde{\theta},x)^{(j,t)} \right) - \ex_{\pi_{\epsilon_{t}}} (\varphi) \right] \overset{d}{\to} N(0, \widetilde{V}_{t}(\varphi)),
    \end{align*}
    under the inductive hypothesis \eqref{the inductive hypothesis}. 
\end{lemma}

Lemma~\ref{lemma: resampling} quantifies the MCSE from multinomial resampling. The asymptotic variance is the sum of \eqref{eq: variance of reweighting} and the variance of $\varphi$ with respect to the posterior $\pi_{\epsilon_t}$. This also indicates that, at $t=T$, the algorithm should stop at the reweighting step to avoid additional randomness from resampling. Nevertheless, we need Lemma~\ref{lemma: reweighting} to complete the mathematical induction for $t$ from the inductive hypothesis at $t-1$.
\end{proof}

\essEstimate*

\begin{proof} 
By the moment assumptions, we have the following two weak convergence results: $\frac{1}{N} \sum_{i=1}^N w_i \overset{p}{\to} \ex_h (w(X)) = 1$ and $\frac{1}{N} \sum_{i=1}^N w_i^2 \overset{p}{\to} \ex_h (w^2(X)) = 1 + \var_h(w(X))$. Then, by the continuous mapping theorem,
\begin{align*}
    \widehat{ESS_N}/N = \frac{\left( \frac{1}{N} \sum_{i=1}^N w_i \right)^2}{\frac{1}{N}\sum_{i=1}^N w_i^2} \overset{p}{\to} \frac{1}{1+\var_h(w(X))}.
\end{align*}
\end{proof}

The ESS estimate is set up for the proof of Theorem~\ref{thm: CI}.

\CI*

\begin{proof} 
    Let $\sigma_t^2$ be the variance of $\varphi$ from the i.i.d. samples \((\theta, x)^{(i,t)}\) drawn from the intermediate private posterior \(\pi_{\epsilon_t}\) and hence $\var_{\pi_{\epsilon_t}} \left[ \frac{1}{N} \sum_{i=1}^N \varphi(\theta, x)^{(i,t)} \right] = \frac{\sigma_t^2}{N}$. Then, consider
\begin{align*}
    &\sqrt{N}\begin{pmatrix}
    \begin{bmatrix} 
    \frac{1}{N} \sum_{i=1}^N \Tilde{w}^{(i,t)}  \varphi \left( (\theta, x)^{(i,t)} \right) \\ 
    \frac{1}{N} \sum_{i=1}^N \Tilde{w}^{(i,t)}   
    \end{bmatrix} 
    -
    \begin{bmatrix} 
    \gamma_t \ex_t [\varphi] \\ 
    \gamma_t 
    \end{bmatrix}
    \end{pmatrix} \\
    &\xrightarrow{d} \mathcal{N} \left( 0_{2}, \Sigma_t = 
    \begin{pmatrix} 
        \var_{h_t} (\Tilde{w}_t \varphi) & \cov_{h_t} (\Tilde{w}_t \varphi, \Tilde{w}_t) \\ 
        \cov_{h_t} (\Tilde{w}_t \varphi, \Tilde{w}_t) & \var_{h_t} (\Tilde{w}_t) 
    \end{pmatrix} 
    \right).
\end{align*}
By first-order Delta method for $f_1(x,y) = \frac{x}{y}$ with gradient $\nabla f_1(x,y) = (\frac{1}{y}, \frac{-x}{y^2})$, we have
\begin{align*}
    \sqrt{N} \left( \hat{\ex}_t [\varphi] - \ex_t [\varphi] \right) 
    \xrightarrow{d} \mathcal{N} \left( 0, \left( \frac{1}{\gamma_t}, \frac{-1}{\gamma_t} \ex_t [\varphi] \right) \Sigma_t \left( \frac{1}{\gamma_t}, \frac{-1}{\gamma_t} \ex_t [\varphi] \right)^{\top} \right),
\end{align*}
and
\begin{align}
    &V_t = \left( \frac{1}{\gamma_t}, \frac{-1}{\gamma_t} \ex_t [\varphi] \right) \Sigma_t \left( \frac{1}{\gamma_t}, \frac{-1}{\gamma_t} \ex_t [\varphi] \right)^{\top} \nonumber \\
    =& \var_{h_t} (w_t \varphi) - 2 \ex_t [\varphi] \cov_{h_t} (w_t \varphi, w_t) + \left( \ex_t [\varphi] \right)^2 \var_{h_t} (w_t) \nonumber \\
    =&\left[ \ex_{h_t} \left( w_t^2 \varphi^2 \right) - \left( \ex_{h_t} (w_t \varphi) \right)^2 \right]
    - 2 \ex_t [\varphi] \left[ \ex_{h_t} \left( w_t^2 \varphi \right) - \ex_{h_t} \left( w_t \varphi \right) \ex_{h_t} \left( w_t \right) \right] + \left( \ex_t [\varphi] \right)^2 \var_{h_t} (w_t) \nonumber \\
    =&\left[ \ex_t \left(w_t \varphi^2 \right) - \left( \ex_t [\varphi] \right)^2 \right] 
    - 2 \ex_t [\varphi] \left[ \cov_t (w_t, \varphi) + \ex_t (w_t) \ex_t [\varphi] - \ex_t [\varphi] \right] + \left( \ex_t [\varphi] \right)^2 \var_{h_t} (w_t) \label{eq: w_t = pi_t / h_t} \\
    =&\left[ \{\ex_t (w_t) \left(\ex_t [\varphi]\right)^2 + 2 \cov_t(w_t, \varphi) \ex_t [\varphi] + \var_t (\varphi) \ex_t (w_t) + r \} - \left( \ex_t [\varphi] \right)^2 \right] \nonumber \\
    & - 2 \ex_t [\varphi] \left[ \cov_t (w_t, \varphi) + \ex_t (w_t) \ex_t [\varphi] - \ex_t [\varphi] \right] + \left( \ex_t [\varphi] \right)^2 \var_{h_t} (w_t) \label{eq: second-order Taylor's approximation} \\
    =& \left( \ex_t [\varphi] \right)^2 \left[ 1 + \var_{h_t} (w_t) -\ex_t (w_t) \right] + \ex_t (w_t) \var_t (\varphi) + r \nonumber \\
    =& \left( 1 + \var_{h_t} (w_t) \right) \var_t (\varphi) + r \label{eq: N/ESS} \\
    =& \frac{N}{ESS_N} \sigma_t^2 + r \nonumber,
\end{align}
where $\ex_t(\cdot)$ and $\cov_t(\cdot, \cdot)$ are integral with respect to $\pi_{\epsilon_t} (\theta, x\mid \sdp)$.

For \eqref{eq: w_t = pi_t / h_t}, we apply the fact $w_t = \frac{\pi_{\epsilon_t}}{h_t}$. For \eqref{eq: second-order Taylor's approximation}, we consider the second-order Taylor's approximation for $f_2(x,y) = x y^2$ with Hessian $H_{f_2}(x,y) = \begin{pmatrix} 0 \ 2y \\ 2y \ 2x \end{pmatrix}$ such that 
\begin{align*}
    \ex (f_2 (T)) = \ex (f_2 (\ex T)) + \frac{1}{2} \text{tr}(\cov(T) H_{f_2}(\ex T)) + r,
\end{align*}
where $r = \ex_t \left[ \left( w_t - \ex_t (w_t) \right) \left( \varphi - \ex_t [\varphi] \right)^2 \right]$ is the remainder for the non-zero third order term. For \eqref{eq: N/ESS}, we note that fact that $\ex_t (w_t^2) = 1 + \var_{h_t} (w_t)$.

Since $V_t = \frac{\sigma_t^2}{ESS_N/N} + r$ and $r = \ex_t \left[ \left( w_t - \ex_t (w_t) \right) \left( \varphi - \ex_t [\varphi] \right)^2 \right]$, a consistent estimate $\hat{V}_t$ is
\begin{align*} 
    N \frac{\sum_{i=1}^N \textbf{w}^{(i,t)} \left[ \varphi \left( (\theta, x)^{(i,t)} \right) - \hat{\ex}_t [\varphi] \right]^2}{\widehat{ESS_N}} + \sum_{i=1}^N \textbf{w}^{(i,t)} \left[\textbf{w}^{(i,t)} -\frac{1}{N} \right] \left[ \varphi \left( (\theta, x)^{(i,t)} \right) - \hat{\ex}_t [\varphi] \right]^2.
\end{align*}

By Slutsky's theorem, we obtain the following CLT:
\begin{align*}
    \sqrt{N}\left(\frac{\hat{\ex}_t [\varphi] - \ex_t [\varphi]}{\sqrt{\hat{V}_t}} \right) \xrightarrow{d} \mathcal{N} \left( 0, 1 \right),
\end{align*}
from which we can build an asymptotic $(1-\alpha)$ confidence interval for $\ex_t[\varphi]$ as
    \begin{align*}
        \hat{\ex}_t[\varphi] \pm z_{1-\frac{\alpha}{2}}\sqrt{\frac{\hat{V}_t(\varphi)}{N}},
    \end{align*}
    where $z_{1-\frac{\alpha}{2}}$ is the $\left(1-\frac{\alpha}{2}\right)$ quantile of a standard normal distribution.
\end{proof}

\bibliography{references} 

@article{dwork2006calibrating,
    author = {Dwork, Cynthia and McSherry, Frank and Nissim, Kobbi and Smith, Adam},
    year = {2006},
    month = {01},
    pages = {265-284},
    title = {Calibrating Noise to Sensitivity in Private Data Analysis},
    volume = {Vol. 3876},
    isbn = {978-3-540-32731-8},
    journal = {Theory of Cryptography},
    COMMENTdoi = {10.1007/11681878_14}
}

@article{vadhan2017complexity,
  title={The complexity of differential privacy},
  author={Vadhan, Salil},
  journal={Tutorials on the Foundations of Cryptography: Dedicated to Oded Goldreich},
  pages={347--450},
  year={2017},
  publisher={Springer}
}

@inproceedings{erlingsson2014rappor,
  title={Rappor: Randomized aggregatable privacy-preserving ordinal response},
  author={Erlingsson, {\'U}lfar and Pihur, Vasyl and Korolova, Aleksandra},
  booktitle={Proceedings of the 2014 ACM SIGSAC conference on computer and communications security},
  pages={1054--1067},
  year={2014}
}

@inproceedings{abadi2016deep,
  title={Deep learning with differential privacy},
  author={Abadi, Martin and Chu, Andy and Goodfellow, Ian and McMahan, H Brendan and Mironov, Ilya and Talwar, Kunal and Zhang, Li},
  booktitle={Proceedings of the 2016 ACM SIGSAC Conference on Computer and Communications Security},
  pages={308--318},
  year={2016}
}

@inproceedings{dwork2006differential,
  title={Differential privacy},
  author={Dwork, Cynthia},
  booktitle={International colloquium on automata, languages, and programming},
  pages={1--12},
  year={2006},
  organization={Springer}
}

@article{ding2017collecting,
  title={Collecting telemetry data privately},
  author={Ding, Bolin and Kulkarni, Janardhan and Yekhanin, Sergey},
  journal={Advances in Neural Information Processing Systems},
  volume={30},
  year={2017}
}

@article{Abowd20222020,
	author = {Abowd, John M. and Ashmead, Robert and Cumings-Menon, Ryan and Garfinkel, Simson and Heineck, Micah and Heiss, Christine and Johns, Robert and Kifer, Daniel and Leclerc, Philip and Machanavajjhala, Ashwin and Moran, Brett and Sexton, William and Spence, Matthew and Zhuravlev, Pavel},
	journal = {Harvard Data Science Review},
	year = {2022},
	month = {jun 24},
	note = {Special Issue 2. https://hdsr.mitpress.mit.edu/pub/7evz361i},
	publisher = {The MIT Press},
	title = {The 2020 {Census} {Disclosure} {Avoidance} {System} {TopDown} {Algorithm}},
}

@inproceedings{murray2006mcmc,
  title={{MCMC} for doubly-intractable distributions},
  author={Murray, Iain and Ghahramani, Zoubin and MacKay, David JC},
  booktitle={Proceedings of the 22nd Annual Conference on Uncertainty in Artificial Intelligence (UAI-06)},
  pages={359--366},
  year={2006},
  organization={AUAI Press}
}

@inproceedings{Bernstein2018,
 author = {Bernstein, Garrett and Sheldon, Daniel R},
 booktitle = {Advances in Neural Information Processing Systems},
 editor = {S. Bengio and H. Wallach and H. Larochelle and K. Grauman and N. Cesa-Bianchi and R. Garnett},
 pages = {},
 publisher = {Curran Associates, Inc.},
 title = {Differentially Private {Bayesian} Inference for Exponential Families},
 COMMENTurl = {https://proceedings.neurips.cc/paper_files/paper/2018/file/08040837089cdf46631a10aca5258e16-Paper.pdf},
 volume = {31},
 year = {2018}
}

@article{bernstein2019differentially,
  title={Differentially private {Bayesian} linear regression},
  author={Bernstein, Garrett and Sheldon, Daniel R},
  journal={Advances in Neural Information Processing Systems},
  volume={32},
  year={2019}
}

@book{Jim1985,
    author = {Berger, James O.},
    address = {New York, N.Y},
    booktitle = {Statistical {Decision} {Theory} and {Bayesian} {Analysis}},
    edition = {2nd ed.},
    isbn = {9781475742862},
    language = {eng},
    publisher = {Springer},
    series = {Springer series in statistics},
    title = {Statistical {Decision} {Theory} and {Bayesian} {Analysis}},
    year = {1985},
}

@article{berger1994overview,
  title={An {Overview} of {Robust} {Bayesian} {Analysis}},
  author={Berger, James O and Moreno, El{\'\i}as and Pericchi, Luis Raul and Bayarri, M Jes{\'u}s and Bernardo, Jos{\'e} M and Cano, Juan A and De la Horra, Juli{\'a}n and Mart{\'\i}n, Jacinto and R{\'\i}os-Ins{\'u}a, David and Betr{\`o}, Bruno and others},
  journal={TEST},
  volume={3},
  number={1},
  pages={5--124},
  year={1994},
  publisher={Springer}
}

@inproceedings{kulkarni2021differentially,
  title={Differentially Private {Bayesian} Inference for Generalized Linear Models},
  author={Kulkarni, Tejas and J{\"a}lk{\"o}, Joonas and Koskela, Antti and Kaski, Samuel and Honkela, Antti},
  booktitle={International Conference on Machine Learning},
  pages={5838--5849},
  year={2021},
  organization={PMLR}
}

@inproceedings{ferrando2022parametric,
  title={Parametric Bootstrap for Differentially Private Confidence Intervals},
  author={Ferrando, Cecilia and Wang, Shufan and Sheldon, Daniel},
  booktitle={International Conference on Artificial Intelligence and Statistics},
  pages={1598--1618},
  year={2022},
  organization={PMLR}
}

@article{wang2018statistical,
  title={Statistical Approximating Distributions under Differential Privacy},
  author={Wang, Yue and Kifer, Daniel and Lee, Jaewoo and Karwa, Vishesh},
  journal={Journal of Privacy and Confidentiality},
  volume={8},
  number={1},
  year={2018}
}

@article{awan2025simulation,
  title={Simulation-based, finite-sample inference for privatized data},
  author={Awan, Jordan and Wang, Zhanyu},
  journal={Journal of the American Statistical Association},
  volume={120},
  number={551},
  pages={1669--1682},
  year={2025},
  publisher={Taylor \& Francis}
}

@article{wang2025differentially,
  title={Differentially Private Bootstrap: New Privacy Analysis and Inference Strategies},
  author={Wang, Zhanyu and Cheng, Guang and Awan, Jordan},
  journal={Journal of Machine Learning Research},
  volume={26},
  number={257},
  pages={1--57},
  year={2025}
}

@inproceedings{smith2011,
author = {Smith, Adam},
year = {2011},
month = {06},
pages = {813-822},
title = {Privacy-preserving Statistical Estimation with Optimal Convergence Rates},
booktitle = {Proceedings of the 43rd Annual ACM Symposium on Theory of Computing (STOC '11)}
}

@article{cai2021cost,
  title={The Cost of Privacy: Optimal Rates of Convergence for Parameter Estimation with Differential Privacy},
  author={Cai, Tony and Wang, Yichen and Zhang, Linjun},
  journal={The Annals of Statistics},
  volume={49},
  number={5},
  pages={2825--2850},
  year={2021},
  publisher={Institute of Mathematical Statistics}
}

@inproceedings{gaboardi2016differentially,
  title={Differentially Private Chi-squared Hypothesis Testing: Goodness of Fit and Independence Testing},
  author={Gaboardi, Marco and Lim, Hyun and Rogers, Ryan and Vadhan, Salil},
  booktitle={International Conference on Machine Learning},
  pages={2111--2120},
  year={2016},
  organization={PMLR}
}

@article{gaboardi2017local,
  title={Local Private Hypothesis Testing: Chi-square Tests},
  author={Gaboardi, Marco and Rogers, Ryan},
  journal={arXiv preprint arXiv:1709.07155},
  year={2017}
}

@article{awan2018ump,
  title={Differentially Private Uniformly Most Powerful Tests for Binomial Data},
  author={Awan, Jordan and Slavkovi{\'c}, Aleksandra},
  journal={Advances in Neural Information Processing Systems},
  volume={31},
  year={2018}
}

@article{awan2020differentially,
  title={Differentially Private Inference for Binomial Data},
  author={Awan, Jordan Alexander and Slavkovic, Aleksandra},
  journal={Journal of Privacy and Confidentiality},
  volume={10},
  number={1},
  year={2020}
}

@article{karwa2017sharing,
  title={Sharing Social Network Data: Differentially Private Estimation of Exponential Family Random-graph Models},
  author={Karwa, Vishesh and Krivitsky, Pavel N and Slavkovi{\'c}, Aleksandra B},
  journal={Journal of the Royal Statistical Society Series C: Applied Statistics},
  volume={66},
  number={3},
  pages={481--500},
  year={2017},
  publisher={Oxford University Press}
}

@inproceedings{schein2019locally,
  title={Locally Private {Bayesian} Inference for Count Models},
  author={Schein, Aaron and Wu, Zhiwei Steven and Schofield, Alexandra and Zhou, Mingyuan and Wallach, Hanna},
  booktitle={International Conference on Machine Learning},
  pages={5638--5648},
  year={2019},
  organization={PMLR}
}

@article{wang2025optimal,
  title={Optimal Debiased Inference on Privatized Data via Indirect Estimation and Parametric Bootstrap},
  author={Wang, Zhanyu and Chang, Arin and Awan, Jordan},
  journal={arXiv preprint arXiv:2507.10746},
  year={2025}
}

@article{awan2026large,
  title={Large-Sample {Bayesian} Approximations for Privatized Data},
  author={Awan, Jordan and Chen, Xi and Molinari, Roberto},
  journal={arXiv preprint arXiv:2604.24817},
  year={2026}
}

@article{Liu1998,
    author = {Jun S. Liu and Rong Chen},
    title = {Sequential {Monte} {Carlo} Methods for Dynamic Systems},
    journal = {Journal of the American Statistical Association},
    volume = {93},
    number = {443},
    pages = {1032--1044},
    year = {1998},
    publisher = {ASA Website},
    COMMENTdoi = {10.1080/01621459.1998.10473765}
}

@article{Carpenter2000,
author = {Carpenter, James and Cliffordy, Peter and Fearnhead, Paul},
year = {2000},
month = {09},
pages = {},
title = {An Improved Particle Filter for Non-linear Problems},
volume = {146},
journal = {IEE Proc. Radar, Sonar Navig.}
}

@article{kong1992,
    author = {Augustine Kong},
    title = {A Note on Importance Sampling using Standardized Weights},
    journal = {Technical Report No.348},
    year = {1992}
}

@article{Geweke1989,
 ISSN = {00129682, 14680262},
 COMMENTurl = {http://www.jstor.org/stable/1913710},
 author = {John Geweke},
 journal = {Econometrica},
 number = {6},
 pages = {1317--1339},
 publisher = {[Wiley, Econometric Society]},
 title = {{Bayesian} Inference in Econometric Models Using {Monte} {Carlo} Integration},
 COMMENTurldate = {2024-10-18},
 volume = {57},
 year = {1989}
}

@article{gong2022exact,
  title={Exact Inference with Approximate Computation for Differentially Private Data via Perturbations},
  author={Gong, Ruobin},
  journal={Journal of Privacy and Confidentiality},
  volume={12},
  number={2},
  year={2022}
}

@article{ju2022data,
  title={Data Augmentation {MCMC} for {Bayesian} Inference from Privatized Data},
  author={Ju, Nianqiao and Awan, Jordan and Gong, Ruobin and Rao, Vinayak},
  journal={Advances in Neural Information Processing Systems},
  volume={35},
  pages={12732--12743},
  year={2022}
}

@article{filippi2013optimality,
  title={On Optimality of Kernels for Approximate {Bayesian} Computation Using Sequential {Monte Carlo}},
  author={Filippi, Sarah and Barnes, Chris P and Cornebise, Julien and Stumpf, Michael PH},
  journal={Statistical Applications in Genetics and Molecular Biology},
  volume={12},
  number={1},
  pages={87--107},
  year={2013},
  publisher={De Gruyter}
}

@article{Beaumont2009,
title = {Adaptive Approximate {Bayesian} Computation},
author = {Beaumont, Mark A. and Cornuet, Jean-Marie and Marin, Jean-Michel and Robert, Christian P.},
year = {2009},
month = dec,
volume = {96(4)},
pages = {983 -- 990},
journal = {Biometrika},
publisher = {Oxford University Press},
}

@article{toni_approximate_2009,
	title = {Approximate {Bayesian} Computation Scheme for Parameter Inference and Model Selection in Dynamical Systems},
	volume = {6},
	issn = {1742-5689},
	COMMENTdoi = {10.1098/rsif.2008.0172},
	language = {eng},
	number = {31},
	journal = {Journal of the Royal Society, Interface},
	author = {Toni, Tina and Welch, David and Strelkowa, Natalja and Ipsen, Andreas and Stumpf, Michael P. H.},
	month = feb,
	year = {2009},
	pmid = {19205079},
	pmcid = {PMC2658655},
	pages = {187--202},
}

@article{chopin2004,
author = {Nicolas Chopin},
title = {Central Limit Theorem for Sequential {Monte} {Carlo} Methods and Its Application to {Bayesian} Inference},
volume = {32},
journal = {The Annals of Statistics},
number = {6},
publisher = {Institute of Mathematical Statistics},
pages = {2385 -- 2411},
year = {2004},
COMMENTdoi = {10.1214/009053604000000698},
COMMENTurl = {https://doi.org/10.1214/009053604000000698}
}

@article{dong_gaussian_2022,
	title = {Gaussian Differential Privacy},
	volume = {84},
	issn = {1369-7412},
	COMMENTurl = {https://doi.org/10.1111/rssb.12454},
	COMMENTdoi = {10.1111/rssb.12454},
	number = {1},
	COMMENTurldate = {2024-05-17},
	journal = {Journal of the Royal Statistical Society Series B: Statistical Methodology},
	author = {Dong, Jinshuo and Roth, Aaron and Su, Weijie J.},
	month = feb,
	year = {2022},
	pages = {3--37}
}

@article{dwork2014algorithmic,
  title={The Algorithmic Foundations of Differential Privacy},
  author={Dwork, Cynthia and Roth, Aaron and others},
  journal={Foundations and Trends{\textregistered} in Theoretical Computer Science},
  volume={9},
  number={3--4},
  pages={211--407},
  year={2014},
  publisher={Now Publishers, Inc.}
}

@inproceedings{bun2016concentrated,
  title={Concentrated differential privacy: Simplifications, extensions, and lower bounds},
  author={Bun, Mark and Steinke, Thomas},
  booktitle={Theory of Cryptography Conference},
  pages={635--658},
  year={2016},
  organization={Springer}
}

@inproceedings{mironov2017renyi,
  title={R{\'e}nyi differential privacy},
  author={Mironov, Ilya},
  booktitle={2017 IEEE 30th computer security foundations symposium (CSF)},
  pages={263--275},
  year={2017},
  organization={IEEE}
}

@article{chaudhuri2011differentially,
  title={Differentially private empirical risk minimization.},
  author={Chaudhuri, Kamalika and Monteleoni, Claire and Sarwate, Anand D},
  journal={Journal of Machine Learning Research},
  volume={12},
  number={3},
  year={2011}
}

@inproceedings{hardt2010geometry,
  title={On the geometry of differential privacy},
  author={Hardt, Moritz and Talwar, Kunal},
  booktitle={Proceedings of the Forty-Second ACM Symposium on Theory of Computing},
  pages={705--714},
  year={2010}
}

@misc{StatisticsCanada2022,
  author       = {Statistics Canada},
  title        = {2021 Census of Population},
  year         = {2022},
  COMMENTurl          = {https://www12.statcan.gc.ca/census-recensement/index-eng.cfm}
}

@book{Chopin_Papaspiliopoulos_2020, series={Springer Series in Statistics}, title={An Introduction to Sequential {Monte} {Carlo}}, ISBN={978-3-030-47845-2}, COMMENTurl={https://books.google.com/books?id=ZZEAEAAAQBAJ}, publisher={Springer International Publishing}, author={Chopin, Nicolas and Papaspiliopoulos, Omiros}, year={2020}, collection={Springer Series in Statistics} }

@book{durrett2019probability,
  title={Probability: Theory and Examples},
  author={Durrett, Richard},
  year={2019},
  publisher={Cambridge University Press}
}

@article{awan2021structure,
    title = {Structure and Sensitivity in Differential Privacy: Comparing K-Norm Mechanisms},
    volume = {116},
    issn = {0162-1459},
    COMMENTurl = {https://doi.org/10.1080/01621459.2020.1773831},
    COMMENTdoi = {10.1080/01621459.2020.1773831},
    number = {534},
    journal = {Journal of the American Statistical Association},
    author = {Awan, Jordan and Slavković, Aleksandra},
    month = apr,
    year = {2021},
    pages = {935--954},
}

@article{steinke2016between,
    title = {Between Pure and Approximate Differential Privacy},
    volume = {7},
    copyright = {Copyright (c) 2017 Thomas Steinke, Jonathan Ullman},
    issn = {2575-8527},
    language = {en},
    number = {2},
    journal = {Journal of Privacy and Confidentiality},
    author = {Steinke, Thomas and Ullman, Jonathan},
    year = {2016},
    note = {Number: 2},
}

@article{agapiou2017importance,
  author  = {Agapiou, Sergios and Papaspiliopoulos, Omiros and Sanz-Alonso, Daniel and Stuart, Andrew M.},
  title   = {Importance Sampling: Intrinsic Dimension and Computational Cost},
  journal = {Statistical Science},
  year    = {2017},
  volume  = {32},
  number  = {3},
  pages   = {405--431}
}

@article{talwar2015lasso,
  title={Nearly optimal private lasso},
  author={Talwar, Kunal and Guha Thakurta, Abhradeep and Zhang, Li},
  journal={Advances in Neural Information Processing Systems},
  volume={28},
  year={2015}
}

@article{del2006sequential,
  title = {Sequential Monte Carlo Samplers},
  author = {Del Moral, Pierre and Doucet, Arnaud and Jasra, Ajay},
  journal = {Journal of the Royal Statistical Society: Series B (Statistical Methodology)},
  volume = {68},
  number = {3},
  pages = {411--436},
  year = {2006},
  publisher = {Wiley-Blackwell},
}

\end{document}